\documentclass[11pt]{article}
\usepackage{fullpage}
\usepackage{graphicx}
\usepackage{amsmath}
\usepackage{amssymb}
\usepackage[boxed,lined]{algorithm2e}

\usepackage{versions}\excludeversion{INUTILE}

\newtheorem{theorem}{Theorem}
\newtheorem{lemma}[theorem]{Lemma}
\newtheorem{corollary}[theorem]{Corollary}
\newenvironment{proof}{\textit{Proof:}}{\hfill$\Box$ \paragraph{} }
\newtheorem{example}{Example}

\providecommand{\pii}{\ensuremath{\pi^{-1}}}  % For \pi inverse

\newcommand{\Oh}[1]
    {\ensuremath{\mathcal{O}\left( {#1} \right)}}
\newcommand{\occ}[2]
    {\ensuremath{|{#2}|_{#1}}}
\newcommand{\access}
    {\ensuremath{\mathsf{access}}}
\newcommand{\rank}
    {\ensuremath{\mathsf{rank}}}
\newcommand{\select}
    {\ensuremath{\mathsf{select}}}
\newcommand{\runs}
    {\ensuremath{\mathsf{runs}}}

\newcommand{\union}
    {\ensuremath{\mathsf{union}}}
\newcommand{\find}
    {\ensuremath{\mathsf{find}}}
\newcommand{\HH}{\mathcal{H}}
\newcommand{\Ho}{\HH_0}
\newcommand{\Hk}{\HH_k}

\newcommand{\libcds}{{\sc Libcds}}

\newcommand{\mFromRaman}{\ensuremath\mu} %\mathsf{Max}}

\newcommand{\mapping}{\ensuremath{{m}}}

\newcommand{\Cs}{\mathcal{C}}

\begin{document}

%\title{Alphabet Partitioning for Compressed Rank/Select \\ and Applications
\title{Efficient Fully-Compressed Sequence Representations
\thanks{Funded in part by Fondecyt Project 1-110066, Chile.
An early version of this article appeared in {\em Proc. 21st Annual 
International Symposium on Algorithms and Computation (ISAAC)}, part II,
pp. 215--326, 2010.}
}

\author{J\'{e}r\'{e}my Barbay\thanks{Dept. of Computer Science, 
        University of Chile. {\tt \{jbarbay|gnavarro|yasha\}@dcc.uchile.cl}.}
\and
    Francisco Claude\thanks{David R. Cheriton School of Computer Science,
University of Waterloo, Canada. {\tt fclaude@cs.uwaterloo.ca}.}
\and
    Travis Gagie\thanks{Dept. of Computer Science, Aalto University, Finland. 
    {\tt travis.gagie@gmail.com}.}
\and
    Gonzalo Navarro$^\dag$
\and
    Yakov Nekrich$^\dag$
}

\date{}
\maketitle

\setcounter{footnote}{0}

\begin{abstract}
  We present a data structure that stores a sequence $s[1..n]$ over
  alphabet $[1..\sigma]$ in $n\Ho(s) + o(n)(\Ho(s){+}1)$ bits, where
  $\Ho(s)$ is the zero-order entropy of $s$.
  This structure supports the queries \access, \rank\ and \select,
  which are fundamental building blocks for many other compressed data 
  structures, in worst-case time $\Oh{\lg\lg\sigma}$
  and average time $\Oh{\lg \Ho(s)}$.
  The worst-case complexity matches the best previous results, yet these had
  been achieved with data structures using $n\Ho(s)+o(n\lg\sigma)$ bits. 
  On highly compressible sequences the $o(n\lg\sigma)$ bits of the redundancy 
  may be significant compared to the the $n\Ho(s)$ bits that encode the data. 
  Our representation, instead, compresses the redundancy as well. Moreover, our 
  average-case complexity is unprecedented.
  
  Our technique is based on partitioning the alphabet into characters of 
  similar frequency. The subsequence corresponding to each group can then be 
  encoded using fast uncompressed representations without harming the overall
  compression ratios, even in the redundancy.
  
  The result also improves upon the best current compressed
  representations of several other data structures.
  For example, we achieve
  $(i)$ compressed redundancy, retaining the best time complexities, for the 
  smallest existing full-text self-indexes; 
  $(ii)$ compressed permutations $\pi$ with times for $\pi()$ and 
  $\pii()$ improved to loglogarithmic; and 
  $(iii)$ the first compressed representation of dynamic collections of 
  disjoint sets.
  We also point out various applications to inverted indexes,
  suffix arrays, binary relations, and data compressors.

  Our structure is practical on large alphabets. Our experiments show that,
  as predicted by theory, it dominates the space/time tradeoff map
  of all the sequence representations, both in synthetic and application
  scenarios.
\end{abstract}

\section{Introduction} \label{sec:introduction}

A growing number of important applications require data
representations that are space-efficient and at the same time support
fast query operations.  In particular, suitable representations of
{\em sequences} supporting a small set of basic operations yield
space- and time-efficient implementations for many other data
structures such as full-text indexes \cite{GGV03,GMR06,FMMN07,NM07},
labeled trees \cite{BHMR07,BGMR07,FLMM09}, binary relations
\cite{BHMR07,BCN10}, permutations \cite{BN09} and two-dimensional
point sets \cite{MN07,BLNS09}, to name a few.

Let $s[1..n]$ be a sequence of characters belonging to alphabet
$[1..\sigma]$.  In this article we focus on the following set of
operations, which is sufficient for many applications:
\begin{description}
\item[\(s.\access (i)\)] returns the $i$th character of sequence $s$,
  which we denote \(s [i]\);
\item[\(s.\rank_a (i)\)] returns the number of occurrences of
  character $a$ up to position $i$ in $s$; and
\item[\(s.\select_a (i)\)] returns the position of the $i$th
  occurrence of $a$ in $s$.
\end{description}

Table~\ref{tab:previous} shows the best sequence representations and
the complexities they achieve for the three queries, where $\Hk(s)$
refers to the $k$-th order empirical entropy of $s$ \cite{Man01}. To
implement the operations efficiently, the representations require some
{\em redundancy} space on top of the $n\Ho(s)$ or $n\Hk(s)$ bits
needed to encode the data.
For example, multiary wavelet trees (row 1) represent $s$ within
zero-order entropy space plus just $o(n)$ bits of redundancy, and
support queries in time $\Oh{1+\frac{\lg\sigma}{\lg\lg n}}$. This is
very attractive for relatively small alphabets, and even constant-time
for polylog-sized ones. For large $\sigma$, however, all the other
representations in the table are exponentially faster, and some even
achieve high-order compression. However, their redundancy is higher,
$o(n\lg\sigma)$ bits. While this is still asymptotically negligible
compared to the size of a plain representation of $s$, on highly
compressible sequences such redundancy is not always negligible
compared to the space used to encode the compressed data. This raises
the challenge of retaining the efficient support for the queries {\em
  while compressing the index redundancy} as well.

In this paper we solve this challenge in the case of zero-order
entropy compression, that is, the redundancy of our data structure is
asymptotically negligible compared to the zero-order compressed text
size (not only compared to the plain text size), plus $o(n)$ bits. The
worst-case time our structure achieves is $\Oh{\lg\lg\sigma}$, which
matches the best previous results for large $\sigma$. Moreover, the
average time is {\em logarithmic on the entropy of the sequence},
$\Oh{\lg \Ho(s)}$, under reasonable assumptions on the query
distribution. This average time complexity is also unprecedented: the
only previous entropy-adaptive time complexities we are aware of come
from Huffman-shaped wavelet trees \cite{GGV03}, which have 
recently been shown capable of achieving $\Oh{1+\frac{\Ho(s)}{\lg\lg n}}$
query time with just $o(n)$ bits of redundancy \cite[Thm.~5]{BN11}.

\begin{table}[t]
  \caption{
    Best previous bounds and our new ones for data structures supporting
    \access, \rank\ and \select.
The space bound of the form $\Hk(s)$ holds for any \(k = o (\lg_\sigma n)\),
and those of the form $(1+\epsilon)$ hold for any constant $\epsilon>0$.
On average $\lg\sigma$ becomes $\Ho(s)$ in our time complexities (see 
Corollary~\ref{cor:partitioning-avg}) and in row 1 \cite[Thm.~5]{BN11}.
  }
\label{tab:previous}
\medskip
\resizebox{\textwidth}{!}
{\begin{tabular}
{c@{\hspace{1ex}}|@{\hspace{2ex}}c@{\hspace{3ex}}c@{\hspace{3ex}}c@{\hspace{3ex}}c@{\hspace{1ex}}c}
                           & space (bits)                                                & \access                               & \rank                                        & \select                                  \\
\hline \\[-1ex]
\cite[Thm.~4]{GRR08}              & \(n \Ho (s) + o(n)\)                              & $\Oh{1+\frac{\lg \sigma}{\lg \lg n}}$ & $\Oh{1+\frac{\lg \sigma}{\lg \lg n}}$        & $\Oh{1+\frac{\lg \sigma}{\lg \lg n}}$    \\[1ex]
\cite[Lem.~4.1]{BHMR07}  & \(n \Ho (s) + o(n\lg \sigma)\)            & $\Oh{\lg\lg\sigma}$    & $\Oh{\lg \lg \sigma}$ & $\Oh{1}$ \\[1ex]
\cite[Cor.~2]{GOR10}	 & \(n \Hk (s) + o(n\lg \sigma)\) & $\Oh{1}$ & $\Oh{\lg\lg\sigma}$ & $\Oh{\lg\lg\sigma}$ \\[1ex] %OJO can be any omega_sigma(1)
\cite[Thm.~2.2]{GMR06} & \((1+\epsilon) n \lg\sigma\) & $\Oh{1}$ & $\Oh{\lg \lg \sigma}$                        & $\Oh{1}$ \\[1ex]
\hline \\[-1ex]
Thm.~\ref{thm:partitioning} & \(n \Ho(s) + o(n)(\Ho(s)+1)\)                               & $\Oh{\lg \lg \sigma}$                 & $\Oh{\lg \lg \sigma}$                        & $\Oh{1}$                                 \\[1ex]
Thm.~\ref{thm:partitioning} & \(n \Ho(s) + o(n)(\Ho(s)+1)\)                               & $\Oh{1}$                              & $\Oh{\lg \lg \sigma}$ & $\Oh{\lg \lg \sigma}$                    \\[1ex]
Cor.~\ref{cor:partitioning-constant} & \((1+\epsilon) n \Ho(s) + o(n)\) & $\Oh{1}$
& $\Oh{\lg \lg \sigma}$                        & $\Oh{1}$ \\[.5ex]
\end{tabular}}
\end{table}

Our technique is described in Section~\ref{sec:partitioning}.  It can
be summarized as partitioning the alphabet into sub-alphabets that
group characters of similar frequency in $s$, storing in a multiary
wavelet tree~\cite{FMMN07} the sequence of sub-alphabet identifiers,
and storing separate sequences for each sub-alphabet, containing the
subsequence of $s$ formed by the characters of that
sub-alphabet. Golynski et al.'s~\cite{GMR06} or Grossi et
al.'s~\cite{GOR10} structures are used for these subsequences,
depending on the tradeoff to be achieved. We show that it is sufficient to
achieve compression in the multiary wavelet tree, while benefitting
from fast operations on the representations of the
subsequences.

The idea of alphabet partitioning is not new. It has been used in
practical scenarios such as fax encoding, JPEG and MPEG formats
\cite{PM92,HPN97}, and in other image coding methods \cite{PINS04},
with the aim of speeding up decompression: only the (short)
sub-alphabet identifier is encoded with a sophisticated (and slow)
method, whereas the sub-alphabet characters are encoded with a simple
and fast encoder (even in plain form). Said \cite{Sai05} gave a more
formal treatment to this concept, and designed a dynamic programming
algorithm to find the optimal partitioning given the desired number of
sub-alphabets, that is, the one minimizing the redundancy with respect
to the zero-order entropy of the sequence. He proved that an optimal
partitioning defines sub-alphabets according to ranges of character
frequencies, which reduces the cost of finding such
partitioning to polynomial time and space (more precisely, quadratic
on the alphabet size).

Our contribution in this article is, on one hand, to show that a
particular way to define the sub-alphabets, according to a
quantization of the logarithms of the inverse probabilities of the
characters, achieves $o(\Ho(s)+1)$ bits of redundancy per character of
the sequence $s$. This value, in particular, upper bounds the coding
efficiency of Said's optimal partitioning method.  On the other hand,
we apply the idea to sequence data structures supporting operations
\access/\rank/\select, achieving efficient support of indexed operations 
on the sequence, not only fast decoding.

We also consider various extensions and applications of our main
result.
In Section~\ref{sec:app-text} we show how our result can be used to
improve an existing text index that achieves $k$-th order
entropy~\cite{FMMN07,BHMR07}, so as to improve its redundancy and
query times. In this way we achieve the first self-index with space
bounded by $n\Hk(s) + o(n)(\Hk(s)+1)$ bits, for any $k=o(\lg_\sigma
n)$, able to count and locate pattern occurrences and extract any
segment of $s$ within the time complexities achieved by its fastest
predecessors. We also achieve new space/time tradeoffs for inverted
indexes and binary relations.
In Sections~\ref{sec:permutations} and~\ref{sec:functions} we show how
to apply our data structure to store a compressed permutation and a
compressed function, respectively, supporting direct and inverse
applications and in some cases improving upon previous results
\cite{BN09,BN11,MRRR03,HKP09}. We describe further applications to
text indexes and binary relations.  In particular, an application of
permutations, at the end of Section~\ref{sec:permutations}, achieves
for the first time compressed redundancy to store function $\Psi$ of
text indexes \cite{GV05,GGV03,Sad03}.
Section~\ref{sec:unionfind} shows how to maintain a dynamic collection
of disjoint sets, while supporting operations \union\ and \find, in
compressed form. This is, to the best of our knowledge, the first
result of this kind.

\section{Related work} \label{sec:related}

\paragraph{\em Sampling.}
A basic attempt to provide \rank\ and \select\ functionality on a
sequence $s[1..n]$ over alphabet $[1..\sigma]$ is to store $s$ in
plain form and the values $s.\rank_a(k\cdot i)$ for all $a \in
[1..\sigma]$ and $i \in [1..n/k]$, where $k \in [1..n]$ is a sampling
parameter.  This yields constant-time \access, $\Oh{k/\lg_\sigma n}$
time \rank, and $\Oh{k/\lg_\sigma n + \lg\lg n}$ time \select\ if we
process $\Theta(\lg_\sigma n)$ characters of $s$ in constant time
using universal tables, and organize the rank values for each
character in predecessor data structures.  The total space is $n
\lg\sigma + \Oh{(n/k)\sigma\lg n}$. For example, we can choose
$k=\sigma\lg n$ to achieve total space $n\lg\sigma + \Oh{n}$ (that is,
the data plus the redundancy space). Within this space we can achieve
time complexity $\Oh{\sigma\lg\sigma}$ for \rank\ and
$\Oh{\sigma\lg\sigma + \lg\lg n}$ for \select.

\paragraph{\em Succinct indexes.}
The previous construction separates the sequence data from the
``index'', that is, the extra data structures to provide fast
\rank\ and \select.  There are much more sophisticated representations
for the sequence data that offer constant-time access to
$\Theta(\lg_\sigma n)$ consecutive characters of $s$ (i.e., just as if
$s$ were stored in plain form), yet achieving $n\Hk(s)+o(n\lg\sigma)$
bits of space, for any $k=o(\lg_\sigma n)$ \cite{SG06,GN06,FV07}. We
recall that $\Hk(s)$ is the $k$-th order empirical entropy of $s$
\cite{Man01}, a lower bound to the space achieved by any statistical
compressor that models the character probabilities using the context
of their $k$ preceding characters, so $0 \le \Hk(s) \le \HH_{k-1}(s)
\le H_0(s) \le \lg\sigma$. Combining such sequence representations
with sophisticated indexes that require $o(n\lg\sigma)$ bits of
redundancy \cite{BHMR07,GOR10} (i.e., they are ``succinct''), we
obtain results like row 3 of Table~\ref{tab:previous}.

\paragraph{\em Bitmaps.}
A different alternative is to maintain one bitmap $b_a[1..n]$ per character 
$a \in [1..\sigma]$, marking with a 1 the positions $i$ where $s[i]=a$. Then 
$s.\rank_a(i) = b_a.\rank_1(i)$ and $s.\select_a(j) = b_a.\select_1(j)$. The 
bitmaps can be represented in compressed form using ``fully indexable 
dictionaries'' (FIDs) \cite{RRR02}, so that they operate in constant time and 
the total space is $n\Ho(s)+\Oh{n}+o(\sigma n)$ bits. Even with space-optimal 
FIDs \cite{Pat08,Pat09}, this space is $n\Ho(s)+\Oh{n}+
\Oh{\frac{\sigma n}{\lg^c n}}$ (and the time is $\Oh{c}$) for any constant 
$c$, which is acceptable only for polylog-sized 
alphabets, that is, $\sigma = \Oh{\textrm{polylog}(n)}$. An alternative is to 
use weaker compressed bitmap representations \cite{GV05,OS07} that can support
$\select_1$ in constant time and $\rank_1$ in time $\Oh{\lg n}$, and yield
an overall space of $n\Ho(s)+\Oh{n}$ bits. This can be considered as a succinct
index over a given sequence representation, or we can note that we can
actually solve $s.\access(i)$ by probing all the bitmaps $b_a.\access(i)$.
Although this takes at least $\Oh{\sigma}$ time, it is a simple illustration of
another concept: rather than storing an independent index on top of the data, 
the data is represented in a way that provides \access, \rank\ and \select\ 
operations with reasonable efficiency.

\paragraph{\em Wavelet trees.}
The wavelet tree~\cite{GGV03} is a structure integrating data and
index, that provides more balanced time complexities.  It is a balanced
binary tree with one leaf per alphabet character, and storing bitmaps
in its internal nodes, where constant-time \rank\ and
\select\ operations are supported. By using FIDs
\cite{RRR02} to represent those bitmaps, wavelet trees achieve
$n\Ho(s) + \Oh{\frac{n\lg\sigma \lg\lg n}{\lg n}}$ bits of space and
support all three operations in time proportional to their height,
$\Oh{\lg\sigma}$. Multiary wavelet trees \cite{FMMN07} replace the
bitmaps by sequences over sublogarithmic-sized alphabets
$[1..\sigma']$, $\sigma'=\Oh{\lg^\epsilon n}$ for $0<\epsilon<1$, in
order to reduce that height. The FID technique is extended to
alphabets of those sizes while retaining constant times.  Multiary
wavelet trees obtain the same space as the binary ones, but their time
complexities are reduced by an $\Oh{\lg\lg n}$ factor. Indeed, if
$\sigma$ is small enough, $\sigma = \Oh{\mathrm{polylog}(n)}$, the
tree height is a constant and so are all the query times. Recently,
the redundancy of multiary (and binary) wavelet trees has been reduced
to just $o(n)$ \cite{GRR08}, which yields the results in the first row
of Table~\ref{tab:previous}.%
\footnote{Because of these good results on polylog-sized alphabets, we
  focus on larger alphabets in this article, and therefore do not
  distinguish between redundancies of the form $o(n)\lg\sigma$ and
  $n\,o(\lg\sigma)$, writing $o(n\lg\sigma)$ for all. See also
  Footnote 6 of Barbay et al.~\cite{BHMR07}.}

\paragraph{\em Huffman-shaped wavelet trees.}
Another alternative to obtain zero-order compression is to give
Huffman shape to the wavelet tree \cite{GGV03}. This structure uses
$n\Ho(s)+o(n\Ho(s))+\Oh{n}$ bits even if the internal nodes use a plain 
representation, using $|b|+o(|b|)$ bits \cite{Cla96,Mun96}, for their 
bitmaps $b$. Limiting the height to $\Oh{\lg\sigma}$ retains the worst-case 
times of the balanced version and also the given space \cite{BN11}. In order 
to reduce the time complexities by
an $\Oh{\lg\lg n}$ factor, we can build multiary wavelet trees over
multiary Huffman trees \cite{Huf52}.  This can be combined with the
improved representation for sequences over small alphabets
\cite{GRR08} so as to retain the $n\Ho(s)+o(n)$ bits of space and
$\Oh{1+\frac{\lg\sigma}{\lg\lg n}}$ worst-case times of balanced
multiary wavelet trees. The interesting aspect of using Huffman-shaped
trees is that, if the \access\ queries distribute uniformly over the
text positions, and the character arguments $a$ to $\rank_a$ and
$\select_a$ are chosen according to their frequency in $s$, then the
average time complexities are $\Oh{1+\frac{\Ho(s)}{\lg\lg n}}$, the
weighted leaf depth. This result \cite[Thm.~5]{BN11} improves upon the
multiary wavelet tree representation \cite{GRR08} in the average case.
We note that this result \cite{BN11} involves $\Oh{\sigma\lg n}$ extra
bits of space redundancy, which is negligible only for $\sigma =
o(n/\lg n)$.

\paragraph{\em Reducing to permutations.}
A totally different sequence representation \cite{GMR06} improves the
times to poly-loglogarithmic on $\sigma$, that is, exponentially
faster than multiary wavelet trees when $\sigma$ is large enough. Yet,
this representation requires again uncompressed space, $n\lg\sigma +
\Oh{\frac{n\lg\sigma}{\lg\lg\sigma}}$.%
\footnote{The representation actually compresses to the $k$-th order
  entropy of a different sequence, not $s$ (A. Golynski, personal
  communication).}  It cuts the sequence into chunks of length
$\sigma$ and represents each chunk using a permutation $\pi$ (which
acts as an inverted index of the characters in the chunk). As both
operations $\pi()$ and $\pii()$ are needed, a representation
\cite{MRRR03} that stores the permutation within
$(1+\epsilon)\sigma\lg\sigma$ bits and computes $\pi()$ in constant
time and $\pii()$ in time $\Oh{1/\epsilon}$ is used. Depending on
whether $\pi$ or $\pii$ is represented explicitly, constant time is
achieved for \select\ or for \access. Using a constant value for
$\epsilon$ yields a slightly larger representation that solves both
\access\ and \select\ in constant time.

Later, the space of this representation was reduced to
$n\Ho(s)+o(n\lg\sigma)$ bits while retaining the time complexities of
one of the variants (constant-time \select) \cite{BHMR07}. In turn,
the variant offering constant-time \access\ was superseded by the
index of Grossi et al.~\cite{GOR10}, which achieves high-order
compression and also improves upon a slower alternative that takes the
same space \cite{BHMR07}. The best current times are either constant
or $\Oh{\lg\lg\sigma}$.  We summarize them in rows 2 to 4.

\bigskip

Our contribution, in rows 5 to 7 of Table~\ref{tab:previous}, is to
retain times loglogarithmic on $\sigma$, as in rows 2 to 4, while
compressing the redundancy space. This is achieved only with space
$\Ho(s)$, not $\Hk(s)$.  We also achieve average times depending on
$\Ho(s)$ instead of $\lg\sigma$.

\section{Alphabet partitioning} \label{sec:partitioning}

Let \(s [1..n]\) be a sequence over effective alphabet $[1..\sigma]$.%
\footnote{By {\em effective} we mean that every character appears in
  $s$, and thus \(\sigma \le n\). In Section~\ref{sec:effective} we
  handle the case of larger alphabets.}  We represent $s$ using an
alphabet partitioning scheme. Our data structure has three components:

\begin{enumerate}
\item A character mapping $\mapping[1..\sigma]$ that separates the
  alphabet into sub-alphabets. That is, \(\mapping\) is the sequence
  assigning to each character $a\in[1..\sigma]$ the sub-alphabet
\[\mapping [a] ~=~ \lceil \lg (n / \occ{a}{s}) \lg n \rceil, \]
where $\occ{a}{s}$ denotes the number of occurrences of character $a$
in $s$; note that $\mapping[a] \leq \left\lceil \lg^2 n \right\rceil$
for any $a \in [1..\sigma]$.
\item The sequence $t[1..n]$ of the sub-alphabets assigned to each
  character in $s$. That is, $t$ is the sequence over
  \(\left[1..\left\lceil \lg^2 n \right\rceil \right]\) obtained from
  $s$ by replacing each occurrence of $a$ by \(\mapping [a]\), namely
  $t[i] = \mapping[s[i]]$.
\item The subsequences $s_\ell[1..\sigma_\ell]$ of characters of each
  sub-alphabet. For \(0 \leq \ell \leq \lceil \lg^2 n\rceil\), let
  $\sigma_\ell = \occ{\ell}{\mapping}$, that is, the number of
  distinct characters of $s$ replaced by $\ell$ in $t$. Then \(s_\ell
  [1..\occ{\ell}{t}]\) is the sequence over $[1..\sigma_\ell]$ defined
  by
\[ s_\ell [t.\rank_\ell(i)] = \mapping.\rank_\ell (s [i]),\]
for all $1 \le i \le n$ such that $t[i] = \ell$.
\end{enumerate}

\begin{example} \label{ex:1}
Let $s = \texttt{"alabar a la alabarda"}$. Then $n=20$ and
$\occ{\mathtt{a}}{s} = 9$, $\occ{\mathtt{l}}{s} =
\occ{\mathtt{'\ '}}{s} = 3$, $\occ{\mathtt{b}}{s} =
\occ{\mathtt{r}}{s} = 2$, and $\occ{\mathtt{d}}{s} = 1$.  Accordingly,
we define the mapping as $\mapping[\mathtt{a}] = 5$,
$\mapping[\mathtt{l}] = \mapping[\mathtt{'\ '}] = 12$,
$\mapping[\mathtt{b}] = \mapping[\mathtt{r}] = 15$, and
$\mapping[\mathtt{d}] = 19$. As this is the effective alphabet, and assuming
that the order is \texttt{"'~',a,b,d,l,r"}, we have
$\mapping = (\mathtt{12,5,15,19,12,15})$. So the sequence of sub-alphabet
identifiers is $t[1..20] =
(\mathtt{5,12,5,15,5,15,12,5,12,12,5,12,5,12,5,15,5,15,19,5})$, and the
subsequences are $s_5 =
(\mathtt{1,1,1,1,1,1,1,1,1})$, $s_{12} = (\mathtt{2,1,1,2,1,2})$,
$s_{15} = (\mathtt{1,2,1,2})$, and $s_{19} = (\mathtt{1})$.
\end{example}

With these data structures we can implement the queries on $s$ as follows:
\begin{eqnarray*}
s.\access (i) & = & m.\select_\ell(s_\ell.\access(t.\rank_\ell(i))),
                        ~\textrm{where}~\ell = t.\access(i); \\[1ex]
s.\rank_a (i) & = & s_\ell.\rank_c (t.\rank_\ell (i)),
                        ~\textrm{where}~\ell = \mapping.\access(a)
                        ~\textrm{and}~ c = \mapping.\rank_\ell(a); \\[1ex]
s.\select_a (i) & = & t.\select_\ell (s_\ell.\select_c (i))
                        ~\textrm{where}~\ell = \mapping.\access(a)
                        ~\textrm{and}~ c = \mapping.\rank_\ell(a).
\end{eqnarray*}

\begin{example}
In the representation of Ex.~\ref{ex:1}, we solve
$s.\access(6)$ by first computing $\ell = t.\access(6) = \mathtt{15}$ and
then 
$m.\select_{15}(s_{15}.\access(t.\rank_{15}(6)))
=m.\select_{15}(s_{15}.\access(2))
=m.\select_{15}(2)
=\mathtt{r}$.
Similarly, to solve $s.\rank_\mathtt{l}(14)$ we compute
$\ell = \mapping.\access(\mathtt{l}) = 12$ and
$c = \mapping.\rank_{12}(\mathtt{l}) = 2$. Then we return
$s_{12}.\rank_2 (t.\rank_{12} (14))
=s_{12}.\rank_2 (6)
=3$.
Finally, to solve $s.\select_\mathtt{r}(2)$, we compute 
$\ell = \mapping.\access(\mathtt{r}) = 15$ and
$c = \mapping.\rank_{15}(\mathtt{r}) = 2$, and return
$t.\select_{15} (s_{15}.\select_2(2))
=t.\select_{15} (4)
=18$.
\end{example}

\subsection{Space analysis}
\label{sec:anal}

Recall that the zero-order entropy of $s[1..n]$ is defined as
\begin{equation} \label{eq:H}
\Ho(s) ~~=~~ \sum_{a\in[1..\sigma]}
\frac{\occ{a}{s}}{n}\lg\frac{n}{\occ{a}{s}}.
\end{equation}
Recall also that, by convexity, $n\Ho(s) \ge (\sigma-1)\lg n +
(n-\sigma+1)\lg \frac{n}{n-\sigma+1}$.  The next lemma gives the key
result for the space analysis.

\begin{lemma} \label{lem:space}
Let $s$, $t$, $\sigma_\ell$ and $s_\ell$ be as defined above. Then
$n\Ho(t) + \sum_\ell |s_\ell|\lg\sigma_\ell \in n\Ho(s) + o(n)$.
\end{lemma}
\begin{proof}
First notice that, for any character $1 \le \ell \le \lceil \lg^2 n
\rceil$ it holds that
\begin{equation} \label{eq:sumclase}
\sum_{c,~\ell=\mapping[c]} \occ{c}{s}
~~=~~ |s_\ell|\,.
\end{equation}
Now notice that, if $\mapping[a]=\mapping[b]=\ell$, then
\begin{eqnarray}
\ell ~~=~~ \lceil \lg(n/\occ{a}{s})\lg n \rceil 
	& = & \lceil \lg(n/\occ{b}{s})\lg n \rceil, \nonumber \\
\hspace*{-3cm} \textrm{therefore \hspace{2cm}} 
\lg (n / \occ{b}{s}) - \lg (n / \occ{a}{s}) 
	& < & 1 / \lg n, \nonumber \\
\hspace*{-3cm} \textrm{and so \hspace{5.5cm}} 
\occ{a}{s}  & < & 2^{1 / \lg n} \occ{b}{s}\,. \label{eq:ineq}
\end{eqnarray}
Now, fix $a$, call $\ell=m[a]$, and sum Eq.~(\ref{eq:ineq}) over all
those $b$ such that $\mapping[b]=\ell$. The second step uses
Eq.~(\ref{eq:sumclase}):
\begin{eqnarray}
\sum_{b,~\ell=\mapping[b]} \occ{a}{s} & < & \nonumber
\sum_{b,~\ell=\mapping[b]}  2^{1/\lg n}\,\occ{b}{s},\\[2ex]
\sigma_\ell\, \occ{a}{s} & < & 2^{1/\lg n}\,|s_\ell|\,, \nonumber \\[1ex]
\sigma_\ell & < & 2^{1/\lg n}\,|s_\ell|/\occ{a}{s}\,. \label{eq:ineq2}
\end{eqnarray}
Since $\sum_a \occ{a}{s} = \sum_\ell |s_\ell| = n$, we have, using
Eq.~(\ref{eq:H}), (\ref{eq:sumclase}), and (\ref{eq:ineq2}),
\begin{eqnarray*}
& & \lefteqn{n \Ho (t) ~+~ \sum_\ell |s_\ell| \lg \sigma_\ell} \\[1ex]
& = & \sum_\ell |s_\ell| \lg (n / |s_\ell|) 
      ~+~ \sum_\ell \sum_{a,~\ell=\mapping[a]} \occ{a}{s} \lg \sigma_\ell \\[1ex]
& < & \sum_\ell \sum_{a,~\ell=\mapping[a]} \occ{a}{s} \lg (n / |s_\ell|) 
      ~+~ \sum_\ell \sum_{a,~\ell=\mapping[a]} \occ{a}{s} \lg \left( 2^{1/
                            \lg n} |s_\ell| / \occ{a}{s} \right)\\[1ex]
& = & \sum_\ell \sum_{a,~\ell=\mapping[a]} \occ{a}{s} \lg (n / \occ{a}{s}) 
      ~+~ \sum_\ell \sum_{a,~\ell=\mapping[a]} \occ{a}{s} / \lg n\\[1ex]
& = & \sum_a \occ{a}{s} \lg (n / \occ{a}{s}) ~+~ n / \lg n \\[1ex]
& \in & n \Ho (s) + o(n)\,.
\end{eqnarray*}
\end{proof}

In other words, if we represent $t$ with \(\Ho (t)\) bits per character
and each $s_\ell$ with \(\lg \sigma_\ell\) bits per character,
we achieve a good overall compression.
Thus we can obtain a very compact representation
of a sequence $s$ by storing a compact representation of $t$ and storing
each $s_{\ell}$ as an ``uncompressed'' sequence over an alphabet
of size $\sigma_{\ell}$.

\subsection{Concrete representation}
\label{sec:repr}

We represent $t$ and $\mapping$ as multiary wavelet trees~\cite{FMMN07};
 we represent each $s_\ell$ as either a multiary wavelet tree or an instance of
Golynski et al.'s~\cite[Thm.~2.2]{GMR06}  \access/\rank/\select\ data structure,
depending on whether \(\sigma_\ell \leq \lg n\) or not.
The wavelet tree for $t$ uses at most
\(n \Ho (t) + \Oh{\frac{n(\lg\lg n)^2}{\lg n}}\) bits and operates
in constant time, because its alphabet size is polylogarithmic
(i.e., $\lceil \lg^2 n \rceil$).
If $s_\ell$ is represented as a wavelet tree, it uses at most
\(|s_\ell| \Ho (s_\ell) + \Oh{\frac{|s_\ell|\lg \sigma_\ell \lg\lg n}{\lg n}}\) bits%
\footnote{This is achieved by using block sizes
of length $\frac{\lg n}{2}$ and not $\frac{\lg |s_\ell|}{2}$, at the price
of storing universal tables of size $\Oh{\sqrt{n}\,\textrm{polylog}(n)}=o(n)$ bits.
Therefore all of our $o(\cdot)$ expressions involving $n$ and other variables will be
asymptotic in $n$.}
and again operates in constant time because $\sigma_\ell \le \lg n$; otherwise it uses at most
\(|s_\ell| \lg \sigma_\ell
     + \Oh{\frac{|s_\ell| \lg \sigma_\ell}{\lg \lg \sigma_\ell}}
\leq |s_\ell| \lg \sigma_\ell
     + \Oh{\frac{|s_\ell| \lg \sigma_\ell}{\lg \lg \lg n}}\) bits
(the latter because $\sigma_\ell > \lg n$).
Thus in either case the space for $s_\ell$ is bounded by
$|s_\ell| \lg \sigma_\ell +\Oh{\frac{|s_\ell|\lg \sigma_\ell}{\lg\lg\lg n}}$ bits.
Finally, since $\mapping$ is a sequence of length $\sigma$ over an alphabet of size
$\lceil \lg^2 n \rceil$, the wavelet tree for $\mapping$ takes
$\Oh{\sigma \lg \lg n}$ bits and also operates in constant time. Because of the convexity property we referred to in the
beginning of this section, $n\Ho(s) \ge (\sigma-1)\lg n$, the space for 
$\mapping$ is
$\Oh{\frac{n\lg\lg n}{\lg n}} \cdot \Ho(s)$.

Therefore we have $n\Ho(t)+o(n)$ bits for $t$, 
$\sum_\ell |s_\ell| \lg \sigma_\ell \left(1+\Oh{\frac{1}{\lg\lg\lg n}}\right)$
bits for the $s_\ell$ sequences, and $o(n)\Ho(s)$ bits for $\mapping$. Using
Lemma~\ref{lem:space}, this adds up to
$n\Ho(s) + o(n)\Ho(s) + o(n)$, 
where the $o(n)$ term is $\Oh{\frac{n}{\lg\lg\lg n}}$.

Using the variant of Golynski et al.'s data structure \cite[Thm.~4.2]{GMR06}, 
that gives constant-time \select, and $\Oh{\lg\lg\sigma}$ time for \rank\ and
\access, we obtain our first result in Table~\ref{tab:previous} (row 4). 
To obtain our second result (row 5), we use instead Grossi et al.'s result 
\cite[Cor.~2]{GOR10}, which gives constant-time \access, and $\Oh{\lg\lg\sigma}$
time for \rank\ and \select. We note that their structure takes space 
$|s_\ell|\Hk(s_\ell)+ \Oh{\frac{|s_\ell|\lg\sigma_\ell}{\lg\lg\sigma_\ell}}$, 
yet we only need this to be at most 
$|s_\ell|\lg\sigma_\ell + \Oh{\frac{|s_\ell|\lg\sigma_\ell}{\lg\lg\lg n}}$.

\begin{theorem} \label{thm:partitioning}
We can store $s[1..n]$ over effective alphabet $[1..\sigma]$ in
\(n \Ho (s) + o(n)(\Ho (s) + 1)\) 
bits and support \access, \rank\ and \select\ queries in
$\Oh{\lg \lg \sigma}$, $\Oh{\lg \lg \sigma}$, and $\Oh{1}$ time,
respectively (variant (i)); or in
$\Oh{1}$, $\Oh{\lg \lg \sigma}$ and
$\Oh{\lg \lg \sigma}$ time, respectively (variant (ii)).
\end{theorem}

We can refine
the time complexity by noticing that the only non-constant times are due to
operating on some sequence $s_\ell$, where the alphabet is of size
$\sigma_\ell < 2^{1/\lg n}|s_\ell|/\occ{a}{s}$, where $a$ is the character
in question, thus $\lg\lg\sigma_\ell = \Oh{\lg\lg(n/\occ{a}{s})}$. If we
assume that the characters $a$ used in queries distribute with the same
frequencies as in sequence $s$ (e.g., \access\ queries refer to randomly chosen
positions in $s$), then the average query time becomes 
$\Oh{\sum_a \frac{\occ{a}{s}}{n} \lg\lg\frac{n}{\occ{a}{s}}} = \Oh{\lg \Ho(s)}$
by the log-sum inequality\footnote{Given $\sigma$ pairs of numbers 
$a_i,b_i>0$, it holds that $\sum a_i \lg\frac{a_i}{b_i} \ge
\left(\sum a_i\right)\lg\frac{\sum a_i}{\sum b_i}$.
%Cover and Thomas Thm 2.7.1
Use $a_i = \occ{i}{s}/n$ and $b_i = - a_i \lg a_i$ to obtain the result.}.

\begin{corollary} \label{cor:partitioning-avg}
The $\Oh{\lg\lg \sigma}$ time complexities in
Theorem~\ref{thm:partitioning} are also
$\Oh{\lg\lg (n/\occ{a}{s})}$, where $a$ stands for $s[i]$ in the 
{\access} query, and for the character argument in the 
$\rank_a$ and $\select_a$ queries. If these characters $a$ distribute on queries
with the same frequencies as $s$, the average time complexity for those
operations is $\Oh{\lg \Ho(s)}$.
\end{corollary}

Finally, to obtain our last result in Table~\ref{tab:previous} we use again
Golynski et al.'s representation \cite[Thm.~4.2]{GMR06}. Given 
$\epsilon |s_\ell| \lg \sigma_\ell$ extra space to store the inverse of a
permutation inside chunks, it answers \select\ queries in time $\Oh{1}$ 
and \access\ queries in time $\Oh{1/\epsilon}$ (these two complexities can be
interchanged), and \rank\ queries in time
$\Oh{\lg\lg\sigma_\ell}$. While we initially considered
$1/\epsilon = \lg\lg\sigma_\ell$ to achieve the main result, using a constant
$\epsilon$ yields constant-time \select\ and \access\ simultaneously.

\begin{corollary} \label{cor:partitioning-constant}
We can store $s[1..n]$ over effective alphabet $[1..\sigma]$ in
\((1+\epsilon) n \Ho (s) + o(n)\) bits, for any constant $\epsilon>0$, and 
support \access, $\rank_a$ and \select\ queries in
$\Oh{1/\epsilon}$, $\Oh{\lg \lg \min(\sigma,n/\occ{a}{s})}$, and $\Oh{1}$ time,
respectively (variant (i)); or in $\Oh{1}$, $\Oh{\lg \lg
\min(\sigma,n/\occ{a}{s})}$, and $\Oh{1/\epsilon}$, respectively (variant
(ii)).
\end{corollary}

\subsection{Handling arbitrary alphabets} \label{sec:effective}

In the most general case, $s$ is a sequence over an alphabet $\Sigma$
that is not an effective alphabet, and $\sigma$ characters from $\Sigma$
occur in $s$.
Let $\Sigma'$ be the set of elements that occur in $s$; we can map
characters from $\Sigma'$ to elements of $[1..\sigma]$ by replacing
each $a\in \Sigma'$ with its rank in $\Sigma'$.
All elements of $\Sigma'$ are stored in the ``indexed dictionary'' (ID) data
structure described by Raman et al.~\cite{RRR02}, so that the
following queries are supported in constant time:
for any $a\in\Sigma'$ its rank in $\Sigma'$ can be found (for any
$a\not\in \Sigma'$ the answer is $-1$);
and for any $i\in [1..\sigma]$ the $i$-th smallest element in $\Sigma'$
can be found. The ID structure uses
$\sigma\lg(e\mFromRaman/\sigma)+o(\sigma)+\Oh{\lg\lg\mFromRaman}$ bits
of space, where
$e$ is the base of the natural logarithm and
$\mFromRaman$ is the maximal element in $\Sigma'$;
the value of $\mFromRaman$ can be specified with additional
$\Oh{\lg\mFromRaman}$ bits.
We replace every element in $s$ by its rank in $\Sigma'$, and the
resulting sequence is stored using Theorem~\ref{thm:partitioning}.
Hence, in the general case the space usage is increased by
$\sigma\lg(e\mFromRaman/\sigma)+o(\sigma)+\Oh{\lg\mFromRaman}$ bits and
the asymptotic time complexity of queries remains unchanged.
Since we are already spending $\Oh{\sigma\lg\lg n}$ bits in our data
structure, this increases the given space only by 
$\Oh{\sigma \lg(\mu/\sigma)}$.

\subsection{Application to fast encode/decode}

Given a sequence $s$ to encode, we can build mapping $\mapping$ from its
character frequencies $\occ{a}{s}$, and then encode each $s[i]$ as the pair
$(m[s[i]], m.\rank_{m[s[i]]}(s[i]))$. 
Lemma~\ref{lem:space} (and some of the discussion that 
follows in Section~\ref{sec:repr}) shows that the overall output size is
$n\Ho(s)+o(n)$ bits if we represent the sequence of pairs 
by partitioning it into three sequences: (1) the left part of the pairs in one 
sequence, using Huffman coding on chunks (see next); (2) the right part of the 
pairs corresponding to values where $\sigma_\ell < \lg n$ in a second sequence, 
using Huffman coding on chunks; (3) the remaining right parts of the pairs, 
using plain encoding in $\lceil \lg \sigma_\ell \rceil$ bits (note 
$\sigma_\ell = m.\rank_\ell(\sigma)$). The Huffman coding on chunks groups
$\frac{\lg n}{4\lg\lg n}$ characters, so that even in the case of the left 
parts, where the alphabet is of size $\lceil \lg^2 n \rceil$, the total length 
of a chunk is at most $\frac{\lg n}{2}$ bits, and hence the Huffman coding
table occupies just $\Oh{\sqrt{n}\lg n}$ bits. 
The redundancy on top of $\Ho(s)$ adds up to $\Oh{\frac{n \lg\lg n}{\lg n}}$ 
bits in sequences (1) and (2) (one bit
of Huffman redundancy per chunk) and $\Oh{\frac{n}{\lg\lg n}}$ in sequence (3) 
(one bit, coming from the ceil function, per $\lg \sigma_\ell > \lg\lg n$
encoded bits).

%yakov:
%For canonical  Huffman coding: it follows from
% (Moffat, Turpin, IEEE Trans on Communic., 1997) that we can decode  any
%canonical prefix code
%in O(t) time if a predecessor of L_{max} numbers can be found in O(t) time.
%If L_{max} =O(w),  where w is the word size and L_{max} is the maximum
%codeword
%length,  we can use the data structure of (Fredman and Willard, J. Comp.
%Syst. Sciences,1993) and decode in O(1) time.
%This observation is also used  in our WADS'09 paper with Travis; in that
%paper we describe it even for adaptive codes.
The overall encoding time is $\Oh{n}$. A pair $(\ell,o)$ is 
decoded as $s[i] = m.\select_\ell (o)$, where after reading $\ell$ we can
compute $\sigma_\ell$ to determine whether $o$ is encoded in sequence (2) or
(3). Thus decoding also takes constant time if we can decode Huffman codes in
constant time. This can be achieved by using canonical codes and limiting the
height of the tree \cite{MT97,GN09}. 

This construction gives an interesting space/time tradeoff with respect to
classical alternatives. Using just Huffman coding yields $\Oh{n}$
encoding/decoding time, but only guarantees $n\Ho(s)+\Oh{n}$ bits of space. 
Using arithmetic coding achieves $n\Ho(s)+\Oh{1}$ bits, but encoding/decoding
is not linear-time. The tradeoff given by our encoding, $n\Ho(s)+o(n)$ bits
and linear-time decoding, is indeed the
reason why it is used in practice in various folklore applications, as
mentioned in the Introduction. In Section~\ref{sec:exp-compr} we
experimentally evaluate these ideas and show they are practical. 
Next, we give more far-fetched applications
of the \rank/\select\ capabilities of our structure, which go much beyond
the mere compression.

\section{Applications to text indexing}
\label{sec:app-text}

Our main result can be readily carried over various types of indexes for
text collections. These include self-indexes for general texts, and positional
and non-positional inverted indexes for natural language text collections.

\subsection{Self-indexes}
\label{sec:higher-order}

A self-index represents a sequence and supports
operations related to text searching on it.
A well-known self-index~\cite{FMMN07} achieves $k$-th order entropy space
by partitioning the Burrows-Wheeler transform~\cite{BW94}
of the sequence and encoding
each partition to its zero-order entropy. Those partitions must support
queries \access\ and \rank. By using Theorem~\ref{thm:partitioning}$(i)$ to
represent such partitions, we achieve the following result, improving previous
ones~\cite{FMMN07,GMR06,BHMR07}.

\begin{theorem} \label{thm:self}
  Let $s[1..n]$ be a sequence over effective alphabet $[1..\sigma]$.
  Then we can represent $s$ using $n\Hk(s) + o(n) (\Hk (s) + 1)$ bits,
  for any $k \le (\delta \lg_\sigma n)-1$ and constant
  $0<\delta<1$, while supporting the following queries:
  $(i)$ count the number of occurrences of a pattern $p[1..m]$ in $s$,
  in time $\Oh{m\lg\lg\sigma}$;
  $(ii)$ locate any such occurrence in time $\Oh{\lg n\lg\lg\lg
    n\lg\lg\sigma}$;
  $(iii)$ extract $s[l,r]$ in time $\Oh{(r-l) \lg\lg\sigma + \lg
    n\lg\lg\lg n\lg\lg\sigma}$.
\end{theorem}

\begin{proof}
To achieve $n\Hk(s)$ space, the Burrows-Wheeler transformed text $s^{bwt}$ is 
partitioned into $r \le \sigma^k$ sequences $s^1 \ldots s^r$ \cite{FMMN07}. 
Since $k \le (\delta\lg_\sigma n)-1$, it follows that $\sigma^{k+1} \le
n^\delta$. The space our Theorem~\ref{thm:partitioning}$(i)$ achieves using 
such a partition is 
$\sum_i |s^i|\Ho(s^i) + (\Ho(s^i)+1)\cdot\Oh{\frac{|s^i|}{\lg\lg\lg|s^i|}}$.
Let $\gamma = (1-\delta)/2$ (so $0<\delta+\gamma<1$ whenever $0<\delta<1$) and 
classify the sequences $s^i$ according to whether $|s^i| < n^\gamma$ (short 
sequences) or not (long sequences). The total space occupied by the short 
sequences can be bounded by $r \cdot \Oh{n^\gamma \lg\sigma} = 
\Oh{n^{\delta+\gamma}} = o(n)$ bits. In turn, the space occupied by the
long sequences can be bounded by
$\sum_i \left(1+\frac{c}{\lg\lg\lg n}\right)\cdot|s^i|\Ho(s^i) + 
 \frac{d\,|s^i|}{\lg\lg\lg n}$ bits, for some constants $c,d$. An argument very
similar to the one used by Ferragina et al.~\cite[Thm.~4.2]{FMMN07} shows
that these add up to 
$\left(1+\frac{c}{\lg\lg\lg n}\right)\cdot n\Hk(s) + 
 \frac{dn}{\lg\lg\lg n}$. Thus the space
is $n\Hk(s) + o(n)(\Hk(s)+1)$. Other structures required by the alphabet
partitioning technique \cite{FMMN07} add $o(n)$ more bits if $\sigma^{k+1}
\le n^\delta$.

The claimed time complexities stem from the
\rank\ and \access\ times on the partitions.
The partitioning scheme \cite{FMMN07} adds
just constant time overheads. Finally, to achieve the claimed locating and 
extracting times we sample one out of every $\lg n \lg\lg\lg n$ text positions.
This maintains our lower-order space term $o(n)$ within 
$\Oh{\frac{n}{\lg\lg\lg n}}$. 
\end{proof}

In case $[1..\sigma]$ is not the effective alphabet we proceed as described 
in Section~\ref{sec:effective}.
Our main improvement compared to Theorem~4.2 of Barbay et al.~\cite{BHMR07}
is that we have compressed the redundancy from $o(n\lg\sigma)$ to 
$o(n)(\Hk(s)+1)$. Our improved locating times, instead, just owe to the denser
sampling, which Barbay et al.\ could also use.

Note that, by using the zero-order representation of Golynski et
al.~\cite[Thm.~4]{GRR08}, we
could achieve even better space, $n\Hk(s)+o(n)$ bits, and time complexities
$\Oh{1+\frac{\lg\sigma}{\lg\lg n}}$ instead of $\Oh{\lg\lg\sigma}$.%
\footnote{One can retain $\lg\lg n$ 
in the denominator by using block sizes depending on $n$ and not on $|s^i|$, 
as explained in the footnote at the beginning of Section~\ref{sec:repr}.}
Such complexities are convenient for not so large alphabets.

%OJO golynski et al use select to move in O(1) time but we cannot since the
%booster would pay loglog n for this select, due to the partitions. Or
%o(sigma n) bits. This is fixed in my paper with Djamal (submitted to ESA11)
%but I guess we should not use that here.

\subsection{Positional inverted indexes}

These indexes retrieve the positions of any word in a text. They
may store the text compressed up to the zero-order
entropy of the {\em word} sequence $s[1..n]$, which allows direct access to
any word. In addition they store the list of the positions where each distinct
word occurs. These lists can be compressed up to a second zero-order entropy
space \cite{NM07}, so the overall space is at least $2n\Ho(s)$. By regarding
$s$ as a sequence over an alphabet $[1..\nu]$ (corresponding here to the
vocabulary), Theorem~\ref{thm:partitioning} represents $s$ within 
$n\Ho(s) + o(n)(\Ho(s)+1)$ bits, which provides state-of-the-art
compression ratios. Variant $(ii)$ supports 
constant-time access to any text word $s[i]$, and access to the $j$th 
entry of the list of any word $a$ ($s.\select_a(j)$) in time $\Oh{\lg\lg\nu}$. 
These two time complexities are exchanged in variant $(i)$, or both can be 
made constant by spending $\epsilon n\Ho(s)$ redundancy for any constant 
$\epsilon>0$ (using Corollary~\ref{cor:partitioning-constant}). 
The length of the inverted lists can be stored within 
$\Oh{\nu\lg n}$ bits (we also need at least this space to store the sequence 
content of each word identifier).

Apart from supporting this basic access to the list of each word, this
representation easily supports operations that are more complex to implement
on explicit
inverted lists \cite{BLOLS09}. For example, we can find the phrases formed
by two words $w_1$ and $w_2$, that appear $n_1$ and $n_2$ times, by finding the 
occurrences of one and verifying the other in the text, in time 
$\Oh{\min(n_1,n_2)\lg\lg\nu}$. Other more sophisticated intersection algorithms
\cite{BLOLS09} can be implemented by supporting operations such as ``find the
position in the list of $w_2$ that follows the $j$th occurrence of word $w_1$''
($s.\rank_{w_2}(s.\select_{w_1}(j))+1$, in time $\Oh{\lg\lg\nu}$) or ``give
the list of word $w$ restricted to the range $[x..y]$ in the collection''
($s.\select_w(s.\rank_w(x-1)+j)$, for $j \ge 1$, until exceeding $y$, in time
$\Oh{\lg\lg\nu}$ plus $\Oh{1}$ per retrieved occurrence).
In Section~\ref{sec:exp-invl} we evaluate this representation in practice.

\subsection{Binary relations and non-positional inverted indexes}
\label{sec:binrels}

Let $R \subseteq L \times O$, where $L=[1..\lambda]$ are called {\em labels}
and $O=[1..\kappa]$ are called {\em objects}, be a binary relation consisting 
of $n$ pairs. Barbay et al.~\cite{BGMR07} represent the relation as follows. 
Let $l_{i_1} < l_{i_2} < \ldots < l_{i_k}$ be the labels related to an object 
$o\in O$. Then we define sequence $s_o = l_{i_1} l_{i_2} \ldots l_{i_k}$. The 
representation for $R$ is the concatenated sequence 
$s = s_1 \cdot s_2 \cdot \ldots \cdot s_\kappa$, of length $n$, and 
the bitmap $b = 1 0^{|s_1|} 1 0^{|s_2|} \ldots 1 0^{|s_\kappa|} 1$, of length 
$n+\kappa+1$.

This representation allows one to efficiently support various queries 
\cite{BGMR07}:
\begin{description}
\item[{\tt table\_access}]:
is $l$ related to $o$?,
	$s.\rank_l(b.\rank_0(b.\select_1(o+1))) > 
	 s.\rank_l(b.\rank_0(b.\select_1(o)))$;
\item[{\tt object\_select}]:
the $i$th label related to an object $o$, 
	$s.\access(b.\rank_0(b.\select_1(o)+i))$; 
\item[{\tt object\_nb}]:
the number of labels an object $o$ is related to, 
	$b.\select_1(o+1)-b.\select_1(o)-1$; 
\item[{\tt object\_rank}]:
the number of labels $< l$ an object $o$ is related to,
	carried out with a predecessor search in 
        $s[b.\rank_0(b.\select_1(o))..b.\rank_0(b.\select_1(o+1))]$, an
        area of length $\Oh{\lambda}$.
	The predecessor data structure requires $o(n)$ bits as it
        is built over values sampled every $\lg^2\lambda$ positions, and the
	query is completed with a binary search;
\item[{\tt label\_select}]:
the $i$th object related to a label $l$, 
	$b.\rank_1(b.\select_0(s.\select_l(i)))$;
\item[{\tt label\_nb}]:
the number of objects a label $l$ is related to,
	$s.\rank_l(n)$. It can also be solved like {\tt object\_nb},
        using a bitmap similar to $b$ that traverses the table label-wise;
\item[{\tt label\_rank}]:
the number of objects $< o$ a label $l$ is related to,
	$s.\rank_l(b.\rank_0(b.\select_1(o)))$. 
\end{description}

Bitmap $b$ can be represented within $\Oh{\kappa\lg\frac{n}{\kappa}} =
o(n) + \Oh{\kappa}$ bits and support all the operations in constant time 
\cite{RRR02}, and its label-wise variant needs $o(n)+\Oh{\lambda}$ bits.
The rest of the space and time complexities depend on how we represent $s$.

Barbay et al.~\cite{BGMR07} used Golynski et al.'s representation for $s$
\cite{GMR06}, so they achieved $n\lg\lambda + o(n\lg\lambda)$ bits of space,
and the times at rows 2 or 3 in Table~\ref{tab:previous} for the
operations on $s$ (later, Barbay et al.~\cite{BHMR07} achieved $n\Hk(s) +
o(n\lg\lambda)$ bits and slightly worse times).
By instead representing $s$ using Theorem~\ref{thm:partitioning},
we achieve compressed redundancy and slightly improve the times. 

To summarize, we achieve $n\Ho(s) + o(n)(\Ho(s)+1) + 
\Oh{\kappa + \lambda}$ bits, and solve 
{\tt label\_nb} and {\tt object\_nb} in constant
time, and {\tt table\_access} and {\tt label\_rank} in time $\Oh{\lg\lg\lambda}$.
For {\tt label\_select}, {\tt object\_select} and {\tt object\_rank} we
achieve times $\Oh{1}$, $\Oh{\lg\lg\lambda}$ and $\Oh{(\lg\lg\lambda)^2}$, 
respectively, or $\Oh{\lg\lg\lambda}$, $\Oh{1}$ and $\Oh{\lg\lg\lambda}$, 
respectively. Corollary~\ref{cor:partitioning-constant} yields a slightly
larger representation with improved times, and a multiary wavelet tree
\cite[Thm.~4]{GRR08} achieves less space and different times; we leave the 
details to the reader.

A non-positional inverted index is a binary relation that associates each
vocabulary word with the documents where it appears. A typical representation
of the lists encodes the differences between consecutive values, achieving
overall space $\Oh{\sum_v n_v \lg\frac{n}{n_v}}$, where word $v$ appears in
$n_v$ documents \cite{WMB99}. In our representation as a binary relation,
it turns out that $\Ho(s) = \sum_v n_v \lg\frac{n}{n_v}$, and thus the
space achieved is comparable to the classical schemes. Within this space,
however, the representation offers various interesting operations apart
from accessing the $i$th element of a list (using {\tt label\_select}),
including support for various list intersection algorithms;
see Barbay et al.~\cite{BGMR07,BHMR07} for more details.

\section{Compressing permutations} \label{sec:permutations}

Barbay and Navarro \cite{BN09} measured the compressibility of a permutation 
$\pi$ in terms of the entropy of the distribution of the lengths of {\em runs} 
of different kinds. Let $\pi$ be covered by $\rho$ runs (using any of the 
previous definitions of runs~\cite{LP94,BN09,Meh79})
of lengths $\runs(\pi) = \langle n_1, \ldots, n_\rho\rangle$.
Then $\HH(\runs (\pi)) = \sum \frac{n_i}{n}\lg\frac{n}{n_i} \le
\lg\rho$ is called the {\em entropy} of the runs (and, because $n_i
\ge 1$, it also holds $n\HH(\runs(\pi)) \ge (\rho-1)\lg n$).
In their most recent variant \cite{BN11} they were able to store $\pi$ in
$2n\HH(\runs(\pi))+o(n)+ \Oh{\rho\lg n}$ 
bits for runs consisting of interleaved sequences of increasing or decreasing
values, and $n\HH(\runs(\pi))+o(n)+\Oh{\rho\lg n}$ 
bits for contiguous sequences of increasing or decreasing values (or,
alternatively, interleaved sequences of consecutive values). 
In all cases they can compute $\pi()$ and $\pi^{-1}()$ in 
$\Oh{\frac{\lg \rho}{\lg\lg n}}$ time, which on average drops to 
$\Oh{1+\frac{\HH(\runs(\pi))}{\lg\lg n}}$ if the queries are 
uniformly distributed in $[1..n]$.

We now show how to use \access/\rank/\select\ data structures to support the 
operations more efficiently while retaining compressed redundancy space.
In general terms, we exchange their $\Oh{\rho\lg n}$ space term by 
$o(n)\,\HH(\runs(\pi))$, and improve their times to $\Oh{\lg\lg\rho}$ in the
worst case, and to $\Oh{\lg\HH(\runs(\pi))}$ on average (again, this is an
improvement only if $\rho$ is not too small).

We first consider interleaved sequences of
increasing or decreasing values as first defined by
Levcopoulos and Petersson \cite{LP94} for adaptive
sorting, and later on for compression~\cite{BN09}, and then give
improved results for more restricted classes of runs.
In both cases we first consider the application of the permutation
$\pi()$ and its inverse, $\pii()$, and later show how to extend the
support to the iterated application of the permutation, $\pi^k()$,
extending and improving previous results~\cite{MRRR03}.

\begin{theorem} \label{thm:runs}
  Let $\pi$ be a permutation on $n$ elements that consists of $\rho$
  interleaved increasing or decreasing runs, of lengths $\runs (\pi)$.
  Suppose we have a data structure that stores a sequence $s[1..n]$
  over effective alphabet $[1..\rho]$ within $\psi(n,\rho,\Ho(s))$ bits,
  supporting queries \access, \rank, and \select\ in time $\tau(n,\sigma)$.
  Then, given its run decomposition, we can store $\pi$ in 
  $2\psi(n,\rho,\HH(\runs(\pi))) + \rho$ bits,
  and perform $\pi()$ and $\pii()$ queries in time $\Oh{\tau(n,\sigma)}$.
\end{theorem}

\begin{proof}
We first replace all the elements of the $r$th run by $r$, for $1\le r\le\rho$.
Let $s$ be the resulting sequence and let $s'$ be $s$ permuted according to 
$\pi$, that is, \(s' [\pi (i)] = s [i]\).
We store $s$ and $s'$ using the given sequence representation, and also store 
$\rho$ bits indicating whether each run is increasing or decreasing. Note that
$\Ho(s) = \Ho(s') = \HH(\runs(\pi))$, which gives the claimed space.

Notice that an increasing run preserves the relative order of the elements of
a subsequence. Therefore, if \(\pi (i)\) is part of an increasing run, then 
\(s'.\rank_{s [i]} (\pi (i)) = s.\rank_{s [i]} (i)\), so
\[\pi (i) = s'.\select_{s [i]} \left( s.\rank_{s [i]} (i) \right).\]
If, instead, \(\pi (i)\) is part of a decreasing run, then \(s'.\rank_{s [i]} (\pi (i))
= s.\rank_{s [i]} (n) +1 - s.\rank_{s [i]} (i)\), so
\[\pi (i) = s'.\select_{s [i]} \left( s.\rank_{s [i]} (n) +1 - s.\rank_{s [i]} (i) \right).\]
A $\pii()$ query is symmetric (exchange $s$ and $s'$ in the formulas).  
Therefore we compute $\pi()$ and $\pii$ with $\Oh{1}$ calls to \access,
\rank, and \select\ on $s$ or $s'$.
\end{proof}

%OJO
%Here we could obtain O(1) pi by using mmphfs to solve s.rank_s[i](i), then
%choose a repr w/constant access for s and another with constant select for
%s'. But this is Djamal's idea, so it should be in another paper.

\begin{example}
Let $\pi = 1,\mathit{8},\mathbf{9},3,\mathit{6},\mathbf{10},5,\mathit{4},
\mathbf{11},7,\mathit{2},\mathbf{12}$ be formed by three runs (indicated by
the different fonts). Then $s = (1,2,3,1,2,3,1,2,3)$ and 
$s'=(1,2,1,2,1,2,1,2,3,3,3,3)$.
\end{example}

By combining Theorem~\ref{thm:runs} with the representations in 
Theorem~\ref{thm:partitioning}, we obtain a result that improves upon previous 
work \cite{BN09,BN11} in time complexity. Note that if the queried positions $i$
are uniformly distributed in $[1..n]$, then all the \access, \rank, and
\select\ queries follow the same character distribution of the runs, and
Corollary~\ref{cor:partitioning-avg} applies.
Note also that the $\rho$ bits are contained in $o(n)\HH(\runs(\pi))$ because
$n\HH(\runs(\pi)) \ge (\rho-1)\lg n$.

\begin{corollary} \label{cor:runs}
  Let $\pi$ be a permutation on $n$ elements that consists of $\rho$
  interleaved increasing or decreasing runs, of lengths $\runs (\pi)$.
  Then, given its run decomposition, we can store $\pi$ in 
  \(2 n \HH (\runs (\pi)) + o(n)(\HH (\runs (\pi)) +
  1)\) bits and perform $\pi()$ and $\pii()$ queries in
  $\Oh{\lg \lg \rho}$ time.
  On uniformly distributed queries the average times are
  $\Oh{\lg \HH(\runs(\pi))}$.
\end{corollary}

\smallskip

The case where the runs are contiguous is handled within around half the 
space, as a simplification of Theorem~\ref{thm:runs}.

\begin{corollary} \label{cor:contruns-gral}
  Let $\pi$ be a permutation on $n$ elements that consists of $\rho$
  contiguous increasing or decreasing runs, of lengths $\runs (\pi)$.
  Suppose we have a data structure that stores a sequence $s[1..n]$
  over effective alphabet $[1..\rho]$ within $\psi(n,\rho,\Ho(s))$ bits,
  supporting queries \access\ and \rank\ in time $\tau_\mathtt{ar}(n,\rho)$, 
  and \select\ in time $\tau_\mathtt{s}(n,\sigma)$.
  Then, given its run decomposition, we can store $\pi$ in 
  $\psi(n,\rho,\HH(\runs(\pi)) + \rho\lg\frac{n}{\rho}+\Oh{\rho}+o(n)$ bits of 
  space, and perform $\pi()$ queries in time $\Oh{\tau_\mathtt{s}(n,\rho)}$ and 
  $\pii()$ queries in time $\Oh{\tau_\mathtt{ar}(n,\rho)}$.
\end{corollary}
\begin{proof}
We proceed as in Theorem~\ref{thm:runs}, yet now sequence $s$ is of the form
$s = 1^{n_1} 2^{n_2} \ldots \rho^{n_\rho}$, and therefore it can be represented
as a bitmap $b=10^{n_1-1} 10^{n_2-1}\ldots 10^{n_\rho-1}1$. The required
operations are implemented as follows: $s.\access(i) = b.\rank_1(i)$,
$s.\rank_{s[i]}(i) = i-b.\select_1(s[i])+1$,
$s.\rank_{s[i]}(n) = b.\select_1(s[i]+1)-b.\select_1(s[i])$,
and $s.\select_a(i) = b.\select_1(a)+i-1$. Those operations are solved in
constant time using a representation for $b$ that takes
$(\rho+1)\lg(e(n+1)/(\rho+1))+o(n)$ bits
\cite{RRR02}. Added to the $\rho$ bits that mark increasing or decreasing
sequences, this gives the claimed space. The claimed time complexities 
correspond to the operations on $s'$, as those in $s$ take constant time.
\end{proof}

Once again, by combining the corollary with representation $(i)$ in 
Theorem~\ref{thm:partitioning}, we 
obtain results that improve upon previous work \cite{BN09,BN11}.
The $\rho\lg\frac{n}{\rho}$ bits are in $o(n)(\HH(\runs(\pi))+1)$ 
because they are $o(n)$ as long as $\rho = o(n)$, and otherwise they are
$\Oh{\rho} = o(\rho\lg n)$, and $(\rho-1)\lg n \le n\HH(\runs(\pi))$.

\begin{corollary} \label{cor:contruns}
  Let $\pi$ be a permutation on $n$ elements that consists of $\rho$
  contiguous increasing or decreasing runs, of lengths $\runs (\pi)$.
  Then, given its run decomposition, we can store $\pi$ in 
  \(n \HH (\runs (\pi)) + o(n)(\HH (\runs (\pi)) + 1)\) bits and 
  perform $\pi()$ queries in time $\Oh{1}$ and $\pii()$ queries in
  time $\Oh{\lg\lg\rho}$ (and $\Oh{\lg\HH(\runs(\pi))}$ on average for 
  uniformly distributed queries).
\end{corollary}

\smallskip

If $\pi$ is formed by interleaved but strictly incrementing ($+1$) or 
decrementing ($-1$) runs, then $\pii$ is formed by contiguous runs,
in the same number and length \cite{BN09}. This gives an immediate 
consequence of Corollary~\ref{cor:contruns-gral}.

\begin{corollary} \label{cor:strict_runs}
  Let $\pi$ be a permutation on $n$ elements that consists of $\rho$
  interleaved strict increasing or decreasing runs, of lengths $\runs (\pi)$.
  Suppose we have a data structure that stores a sequence $s[1..n]$
  over effective alphabet $[1..\rho]$ within $\psi(n,\rho,\Ho(s))$ bits,
  supporting queries \access\ and \rank\ in time $\tau_\mathtt{ar}(n,\rho)$, 
  and \select\ in time $\tau_\mathtt{s}(n,\sigma)$.
  Then, given its run decomposition, we can store $\pi$ in 
  $\psi(n,\rho,\HH(\runs(\pi))) + \rho\lg\frac{n}{\rho}+\Oh{\rho}+o(n)$ bits of 
  space, and perform $\pi()$ queries in time $\Oh{\tau_\mathtt{ar}(n,\rho)}$ 
  and $\pii()$ queries in time $\Oh{\tau_\mathtt{s}(n,\sigma)}$.
\end{corollary}

\smallskip

For example we can achieve the same space of Corollary~\ref{cor:contruns}, yet
with the times for $\pi$ and $\pii$ reversed.
Finally, if we consider runs for $\pi$ that are both contiguous and 
incrementing or decrementing, then  so are the runs of $\pii$.
Corollary~\ref{cor:contruns-gral} can be further simplified as both $s$ and $s'$
can be represented with bitmaps.

\begin{corollary} \label{cor:contstrictruns}
  Let $\pi$ be a permutation on $n$ elements that consists of $\rho$
  contiguous and strict increasing or decreasing runs, of lengths $\runs (\pi)$.
  Then, given its run decomposition, we can store $\pi$ in 
  $2\rho\lg\frac{n}{\rho}+\Oh{\rho}+o(n)$ bits, and perform $\pi()$ 
  and $\pii()$ in $\Oh{1}$ time.
\end{corollary}

\smallskip

We now show how to achieve exponentiation, $\pi^k(i)$ or $\pi^{-k}(i)$,
within compressed space.
Munro et al.~\cite{MRRR03} reduced the problem of supporting
exponentiation on a permutation $\pi$ to the support of the direct and
inverse application of another permutation, related but with quite
distinct runs than $\pi$. Combining it with any of our results does yield 
compression, but one where the space depends on the lengths of both the 
runs and cycles of $\pi$.
The following construction, extending the technique by Munro et 
al.~\cite{MRRR03}, retains the compressibility in terms of the runs 
of $\pi$, which is more natural. 
It builds an index that uses small additional space to
support the exponentiation, thus allowing the compression of the main
data structure with any of our results.

\begin{theorem} \label{thm:exponent} Suppose we have a representation
  of a permutation $\pi$ on $n$ elements that supports queries $\pi()$ 
  in time \(\tau^+\) and queries $\pii()$ in time $\tau^-$. 
  Then for any \(t \leq n\), we can build a data structure that takes 
  $\Oh{(n / t) \lg n}$ bits and, used in conjunction with operation $\pi()$
  or $\pii()$, supports $\pi^k()$ and $\pi^{-k}()$ queries in 
  $\Oh{t\,\min(\tau^+,\tau^-)}$ time.
\end{theorem}

\begin{proof}
The key to computing $i'=\pi^k(i)$ is to discover that $i$ is in a cycle of length
$\ell$ and to assign it a position $0 \le j < \ell$ within its cycle (note
$j$ is arbitrary, yet we must operate consistently once it is assigned).
Then $\pi^k(i)$ lies in the same cycle, at position 
$j' = (j+k~\textrm{mod}~\ell)$, hence $\pi^k(i) = \pi^{j'-j}(i)$ or
$\pi^{j'+\ell-j}(i)$. Thus all we need is to 
find out $j$ and $\ell$, compute $j'$, and finally find the position $i'$ in
$\pi$ that corresponds to the $j'$th element of the cycle.

We decompose $\pi$ into its cycles and, for every cycle of length at least $t$,
store the cycle's length $\ell$ and an array containing the position $i$ in 
$\pi$ of every $t$th element in the cycle. Those positions $i$ are called 
`marked'.  We also store a binary sequence $b[1..n]$, so that $b[i]=1$ iff $i$ 
is marked. For each marked element $i$ we record to which cycle $i$ belongs and 
the position $j$ of $i$ in its cycle.

To compute \(\pi^k (i)\), we repeatedly apply $\pi()$ at most $t$ times until 
we either loop or find a marked element. In the first case, we have found $\ell$,
so we can assume $j=0$, compute $j' < \ell \le t$, and apply $\pi()$ at most $t$ 
more times to find \(i'=\pi^{j'}(i)=\pi^k (i)\) in the loop. If we reach a 
marked element, instead, we have stored 
the cycle identifier to which $i$ belongs, as well as $j$ and $\ell$. Then we 
compute $j'$ and know that the previous marked position is 
$j^* = t \cdot \lfloor j'/t \rfloor$. The corresponding position $i^*$ is found
at cell $j^*/t$ of the array of positions of marked elements, and we finally
move from $i^*$ to $i'$ by applying $j'-j^* \le t$ times operation $\pi()$,
$i = \pi^{j'-j^*}(i^*)=\pi^k(i)$.
A $\pi^{- k}$ query is similar (note that it does not need to use $\pii()$ as
we can always move forward). Moreover, we can also proceed using $\pi^{-1}()$
instead of $\pi()$, whichever is faster, to compute both $\pi^k()$ and 
$\pi^{-k}()$.

The space is $\Oh{(n/t)\lg n}$ both for the samples and for a compressed
representation of bitmap $b$. Note that we only compute \rank\ at the
positions $i$ such that $b[i]=1$. Thus we can use the ID structure 
\cite{RRR02}, which uses $\Oh{(n/t)\lg t}$ bits.
\end{proof}

\subsection {Application to self-indexes}

These results on permutations apply to a second family of self-indexes, which
is based on the representation of the so-called $\Psi$ function 
\cite{GV05,GGV03,Sad03}. Given the suffix array $A[1..n]$ of sequence $s[1..n]$
over alphabet $[1..\sigma]$, $\Psi$ is defined as $\Psi(i) = 
A^{-1}[(A[i]~\textrm{mod}~n)+1]$. Counting, locating, and extracting is carried
out through permutation $\Psi$, which replaces $s$ and $A$. It is known 
\cite{GV05} that $\Psi$ contains $\sigma$ contiguous increasing runs so that 
$\HH(\runs(\Psi)) = \Ho(s)$, which allows for its compression. Grossi et 
al.~\cite{GGV03} represented $\Psi$ within $n\Hk(s)+\Oh{n}$ bits,
while supporting operation $\Psi()$ in constant time, or within $n\Hk(s)+
o(n\lg\sigma)$ while supporting $\Psi()$ in time $\Oh{\lg\sigma}$. 
By using Corollary~\ref{cor:contruns}, we can achieve the
unprecedented space $n\Ho(s) + o(n)(\Ho(s)+1)$ and support $\Psi()$ in
constant time. In addition we can support the inverse $\Psi^{-1}()$ in time 
$\Oh{\lg\lg\sigma}$. Having both $\Psi()$ and $\Psi^{-1}()$ allows for
bidirectional indexes \cite{RNOM09}, which can for example display a snippet 
around any occurrence found without the need for any extra space for sampling. 
Our construction of 
Theorem~\ref{thm:exponent} can be applied on top of any of those 
representations so as to support operation $\Psi^k()$, which is useful for
example to implement compressed suffix trees, yet the particularities of
$\Psi$ allow for sublogarithmic-time solutions \cite{GGV03}. Note also that
using Huffman-shaped wavelet trees to represent the permutation \cite{BN11} 
yields even less space, $n\Ho(s) + o(n) + \Oh{\sigma\lg n}$ bits, and the 
time complexities are relevant for not so large alphabets.

\section{Compressing functions} \label{sec:functions}

Hreinsson, Kr{\o}yer and Pagh~\cite{HKP09} recently showed how, given a
domain \(X = \{x_1, x_2, \ldots, x_n\} \subset \mathbb{N}\) of numbers that
fit in a machine word, they can represent any \(f: X \rightarrow
[1..\sigma]\) in compressed form and provide constant-time evaluation.
Let us identify function $f$ with the sequence of values
$f[1..n] = f(x_1) f(x_2) \ldots f(x_n)$. Then their representation uses
at most \((1 + \epsilon) n \Ho (f) + \Oh{n} 
%n \min (p_{\max} + 0.086, 1.82 (1 - p_{\max})) 
+ o (\sigma)\) bits, for any constant \(\epsilon > 0\).
%and $p_{\max}$ is the relative frequency of the most common function value.
%
We note that this bound holds even when $\sigma$ is much larger than $n$. 

In the special case where \(X = [1..n]\) and \(\sigma = o(n)\), 
we can achieve constant-time evaluation and a better space
bound using our sequence representations. Moreover, we can support extra
functionality such as computing the pre-image of an element.
A first simple result is obtained by representing $f$ as a sequence.

\begin{lemma} \label{lem:string_function}
Let \(f: [1..n] \rightarrow [1..\sigma]\) be a function.
We can represent $f$ using $n\Ho(f) + o(n)(\Ho(f)+1)+\Oh{\sigma}$ bits and
compute $f(i)$ for any $i\in[1..n]$ in $\Oh{1}$ time, and any element of 
$f^{-1}(a)$ for any $a\in[1..\sigma]$ in time $\Oh{\lg\lg\sigma}$, or vice 
versa.  Using more space, $(1+\epsilon)\Ho(f) + o(n)$ bits for any constant 
$\epsilon>0$, we support both queries in constant time.
The size $|f^{-1}(a)|$ is always computed in $\Oh{1}$ time.
\end{lemma}
\begin{proof}
We represent sequence $f[1..n]$ using Theorem~\ref{thm:partitioning} or
Corollary~\ref{cor:partitioning-constant},
so $f(i) = f.\access(i)$ and the $j$th element of $f^{-1}(a)$ is
$f.\select_a(j)$. To compute $|f^{-1}(a)|$ in constant time we store a
binary sequence
$b = 10^{|f^{-1}(1)|}10^{|f^{-1}(2)|}1\ldots 10^{|f^{-1}(\sigma)|}1$, so
that $|f^{-1}(a)| = b.\select_1(a+1)-b.\select_1(a)-1$. The space is the one
needed to represent $s$ plus $\Oh{\sigma\lg\frac{n}{\sigma}}$ bits to 
represent $b$ using an ID 
\cite{RRR02}. This is $o(n)$ if $\sigma=o(n)$,
and otherwise it is $\Oh{\sigma}$. This extra space is also necessary
because $[1..\sigma]$ may not be the effective alphabet of sequence $f[1..n]$
(if $f$ is not surjective).
\end{proof}

Another source of compressibility frequently arising in real-life functions
is nondecreasing or nonincreasing runs. Let us start by allowing interleaved 
runs. Note that in this case $\HH(\runs(f)) \le \Ho(f)$, where 
equality is achieved if we form runs of equal values only. 

\begin{theorem} \label{thm:function}
  Let \(f: [1..n] \rightarrow [1..\sigma]\) be a 
  function such that sequence $f[1..n]$ consists of
  $\rho$ interleaved non-increasing or non-decreasing runs.
  Then, given its run decomposition,
   we can represent $f$ in \(2 n \HH (\runs (f)) + o (n)(\HH (\runs (f)) +
  1) + \Oh{\sigma}\) bits and compute $f(i)$ for any $i \in [1..n]$, and
  any element in $f^{-1}(a)$ for any $a \in [1..\sigma]$,
  in time $\Oh{\lg\lg\rho}$.
  The size $|f^{-1}(a)|$ is computed in $\Oh{1}$ time.
\end{theorem}

\begin{proof}
We store function $f$ as a combination of the permutation $\pi$ that
stably sorts the values $f(i)$, plus the binary sequence $b$ of
Lemma~\ref{lem:string_function}.
Therefore, it holds
\[f (i) = b.\rank_1 (b.\select_0 (\pii(i))).\] 
Similarly, the $j$th element of $f^{-1}(a)$ is 
\[ \pi(b.\rank_0(b.\select_1(a))+j). \]
Since $\pii$ has the same runs as $f$ (the runs in $f$ can have equal
values but those of $\pii$ cannot), we can represent $\pii$ using
Corollary~\ref{cor:runs} to obtain the claimed time and space complexities.
\end{proof}

\begin{example}
Let $f[1..9] = (1,3,2,5,4,9,8,9,8)$. The odd positions form an increasing
run $(1,2,4,8,8)$ and the even positions form $(3,5,9,9)$. The permutation
$\pi$ sorting the values is $(1,3,2,5,4,7,9,6,8)$, and its inverse is
$\pii=(1,3,2,5,4,8,6,9,7)$. The bitmap $b$ is $101010101011100100$.
\end{example}

If we consider only contiguous runs in $f$, we obtain the following result by 
representing $\pii$ with Corollary~\ref{cor:contruns}.
Note the entropy of contiguous runs is no longer upper bounded by $\Ho(f)$.

\begin{corollary} \label{cor:contiguous_function}
Let \(f: [1..n] \rightarrow [1..\sigma]\) be a function, where 
sequence $f$ consists of $\rho$ contiguous non-increasing or
non-decreasing runs. Then, given its run decomposition, we can represent $f$ 
in \(n \HH (\runs (f)) + o(n)(\HH (\runs (f)) + 1) +\Oh{\sigma}\) bits, and 
compute any $f(i)$ in $\Oh{1}$ time, as well as 
retrieve any element in $f^{-1}(a)$ in time $\Oh{\lg\lg\rho}$. 
The size $|f^{-1}(a)|$ can be computed in $\Oh{1}$ time.
\end{corollary}

In all the above results we can use Huffman-shaped wavelet trees 
\cite{BN11} to obtain an alternative space/time tradeoff. We leave the
details to the reader.

\subsection{Application to binary relations, revisited} 

Recall Section~\ref{sec:binrels}, where we represent a binary relation in
terms of a sequence $s$ and a bitmap $b$.
By instead representing $s$ as a function, we can capture another source of
compressibility, and achieve slightly different time complexities.
Note that $\Ho(s)$ corresponds to the distribution of the number $o_i$
of objects associated with a label $i$, let us call it
 $\HH_\mathrm{lab} = \Ho(s) = \sum \frac{o_i}{n}\lg\frac{n}{o_i}$.
On the other hand, if we regard the contiguous increasing runs of $s$, the 
entropy corresponds to the distribution of the number $l_i$ of labels 
associated with an object $i$, let us call it
$\HH_\mathrm{obj} = \HH(\runs(s)) = \sum \frac{l_i}{n}\lg\frac{n}{l_i}$.

While Section~\ref{sec:binrels} compresses $B$ in terms of $\HH_\mathrm{lab} =
\Ho(s)$, we can use Corollary~\ref{cor:contiguous_function} to achieve
$n\HH_\mathrm{obj} + o(n)(\HH_\mathrm{obj}+1) + \Oh{\kappa + \lambda}$ 
bits of space. Since $f.\access(i) = f(i)$ and $f.\select_a(j)$
is the $j$th element of $f^{-1}(a)$, this representation solves 
{\tt label\_nb}, {\tt object\_nb} and {\tt object\_select} in constant time, 
and {\tt label\_select} and {\tt object\_rank} in time $\Oh{\lg\lg\lambda}$.
Operations {\tt label\_rank} and {\tt table\_access} require $f.\rank$,
which is not directly supported. The former can be solved in time
$\Oh{\lg\lg\lambda \lg\lg\kappa}$ as a predecessor search in $\pi$
(storing absolute samples every $\lg^2\kappa$ positions), and the 
latter in time $\Oh{\lg\lg\lambda}$ as the difference between two 
{\tt object\_rank} queries.
We can also achieve $n\HH_\mathrm{obj} + o(n) + \Oh{\kappa+\lambda}$
bits using Huffman-shaped wavelet trees; we leave the details to the reader.

\section{Compressing dynamic collections of disjoint sets} \label{sec:unionfind}

Finally, we now give what is, to the best of our knowledge, the first result 
about storing a compressed collection of disjoint sets while supporting 
operations \union\ and \find~\cite{TvL84}.  The key point in the next theorem 
is that, as the sets in the collection $C$ are merged, our space bound shrinks 
with the zero-order entropy of the distribution of the function $s$ that 
assigns elements to sets in $C$. We define $\HH(C) = 
\sum\frac{n_i}{n}\lg\frac{n}{n_i} \le \lg |C|$,
where $n_i$ are the sizes of the sets, which add up to $n$.

\begin{theorem} \label{thm:disjoint}
  Let $C$ be a collection of disjoint sets whose union is \([1..n]\).
  For any \(\epsilon > 0\), we can store $C$ in \((1 +
  \epsilon) n \HH (C) + \Oh{|C| \lg n} + o (n) \)
  bits and perform any sequence of $r$ \union\ and \find\
  operations in $\Oh{r\alpha(n) + (1/\epsilon)n\lg \lg n}$ total time,
  where $\alpha(n)$ is the inverse Ackermann's function.
\end{theorem}

\begin{proof}
We first use
Theorem~\ref{thm:partitioning} to store the sequence \(s [1..n]\) in which
\(s [i]\) is the representative of the set containing $i$.  We then store the 
representatives in a standard disjoint-set data structure $D$~\cite{TvL84}. 
Since $\Ho(s) = \HH(C)$, our data
structures take \(n \HH (C) + o (n)(\HH (C) + 1) + \Oh{|C| \lg n}\)
bits.  We can perform a query \(\find (i)\) on $C$ by performing
\(D.\find (s [i])\), and perform a \(\union (i, j)\)
operation on $C$ by performing \(D.\union(D.\find(s[i]), D.\find(s[j]))\).

As we only need \access\ functionality on $s$, we use a simple variant of
Theorem~\ref{thm:partitioning}. We support only \rank\ and \select\ on 
the multiary wavelet tree that represents sequence $t$, and store the $s_\ell$
subsequences as plain arrays. The mapping $\mapping$ is of length $|C|$, so it
can easily be represented in plain form to support constant-time operations,
within $\Oh{|C|\lg n}$ bits. This yields constant time $\access$, and therefore
the cost of the $r$ \union\ and \find\ operations is $\Oh{r \alpha(n)}$
\cite{TvL84}.

For our data structure to shrink as we merge sets, we keep
track of \(\HH (C)\) and, whenever it shrinks by a factor of \(1 +
\epsilon\), we rebuild our entire data structure on the updated values
$s[i] \leftarrow \find(s[i])$. First, note that all those \find\ operations
take $\Oh{n}$ time because of path-compression~\cite{TvL84}: Only the first
time one accesses a node $v\in D$ it may occur that the representative is not
directly $v$'s parent. Thus the overall time can be split into $\Oh{n}$ time
for the $n$ instructions $\find(s[i])$ plus $\Oh{n}$ for the $n$ times a node
$v \in D$ is visited for the first time.

Reconstructing the structure of Theorem~\ref{thm:partitioning} also takes 
$\Oh{n}$ time. The plain structures for $\mapping$ and $s_\ell$ are easily
built in linear time, and so is the multiary wavelet tree supporting \rank\ and
\access~\cite{FMMN07}, as it requires just tables of sampled counters.

Since \(\HH (C)\) is always less than \(\lg n\), we rebuild only 
$\Oh{\lg_{1+\epsilon} \lg n}=\Oh{(1/\epsilon)\lg\lg n}$ times. Thus the
overall cost of rebuilding is $\Oh{(1/\epsilon)n\lg\lg n}$. This completes
our time complexity.

Finally, the space term $o(n)\HH(C)$ is
absorbed by $\epsilon \HH(C)$ by slightly adjusting $\epsilon$, and
this gives our final space formula.
\end{proof}

\section{Experimental results}
\label{sec:exper}

In this section we explore the performance of our structure in practice. We 
first introduce, in Section~\ref{sec:dense}, a simpler and more practical 
alphabet partitioning scheme we call ``dense'', which experimentally performs 
better than the 
one we describe in Section~\ref{sec:partitioning}, but on which we could not 
prove useful space bounds. Next, in Section~\ref{sec:exp-compr} we study the 
performance of both alphabet partitioning methods, as well as the optimal one 
\cite{Sai05}, in terms of compression ratio and decompression performance. 
Given the results of these experiments, we continue only with our dense
partitioning for the rest of the section.

In Section~\ref{sec:exp-seqs} we compare our new sequence representation
with the state of the art, considering the tradeoff between space and time
of operations $\rank$, $\select$, and $\access$. Then,
Sections~\ref{sec:exp-invl}, \ref{sec:exp-ssa}, and \ref{sec:exp-graph} 
compare the same data structures on different real-life applications of 
sequence representations. In the first, the operations are used to emulate an
inverted index on the compressed sequence using (almost) no
extra space. In the second, they are used to emulate self-indexes for text
\cite{NM07}. In the third, they provide access to direct and reverse neighbors 
on graphs represented with adjacency lists.

The machine used for the experiments has an
Intel\textsuperscript{\textregistered}
Xeon\textsuperscript{\textregistered} E5620 at $2.40$GHz, $94$GB of
RAM. We did not use multithreading in our implementations; times
are measured using only one thread, and in RAM.
The operating system is Ubuntu 10.04, with
kernel 2.6.32-33-server.x86\_64. The code was compiled using GNU/GCC
version 4.4.3 with optimization flags \verb|-O9|.

Our code is available in \libcds\ version $1.0.10$, downloadable
from \verb|http://libcds.recoded.cl/|.

\subsection{Dense alphabet partitioning}
\label{sec:dense}

Said \cite{Sai05} proved that an optimal assignment to sub-alphabets must
group consecutive symbols once sorted by frequency. A simple alternative
to the partitioning scheme presented in Section~\ref{sec:partitioning}, and
that follows this optimality principle, is to make mapping $\mapping$ 
group elements into consecutive chunks of doubling size, that
is, $\mapping[a] = \lfloor \lg r(a) \rfloor$, where $r(a)$ is the
rank of $a$ according to its frequency. 
The rest of the scheme to define $t[1..n]$ and the sequences
$s_\ell[1..\sigma_\ell]$ is as in Section~\ref{sec:partitioning}. The
classes are in the range $0 \le \ell \le \lfloor \lg\sigma \rfloor$,
and each element in $s_\ell$ is encoded in $\ell$ bits. As we use all the
available bits of each symbol in sequences $s_\ell$ (except possibly in the
last one), we call this scheme {\em dense}.

We show that this scheme is not much worse than the one proposed in
Section~\ref{sec:partitioning} (which will be called {\em sparse}). 
First consider the total number of bits we use to encode the sequences 
$s_\ell$, $\sum_\ell \ell |s_\ell| = \sum_a |s|_a \lfloor \lg r(a) \rfloor$.
Since $|s|_a \le n / r(a)$ because the symbols are sorted by decreasing
frequency, it holds that $r(a) \le n / |s|_a$ and 
$\sum |s|_a \lfloor \lg r(a) \rfloor \le \sum |s|_a \lg (n / |s|_a) = n\Ho(s)$.
Now consider the number of bits we use to encode 
$t = \lfloor \lg r(s[1]) \rfloor , \ldots, \lfloor \lg r(s[n])) \rfloor$.  
We could store each element $\lfloor \lg r(s[i]) \rfloor$ of $t$ in 
$2 \lfloor \lg (\lfloor \lg r(s[i]) \rfloor + 1) \rfloor - 1$ bits using
$\gamma$-codes \cite{WMB99}, and such encoding would be lower bounded by 
$n\Ho(t)$. Thus
$   n\Ho(t) 
\le 2 \sum_i \lg\lg (r(s[i])+1)
\le 2 \sum_a |s|_a \lg\lg (n/|s|_a+1)
= O(n(\lg\Ho(s) + 1)
= o(n H_0 (s)) + O(n)$
(recall Section~\ref{sec:repr}). It follows that the total encoding length is
$n\Ho(s) + \Oh{n\lg\Ho(s)} = n\Ho(s)+o(n\Ho(s)) + \Oh{n}$ bits.

Apart from the pretty tight upper bound, it is not evident whether this scheme 
is more or less efficient than the sparse encoding. Certainly the dense scheme 
uses the least possible number of classes (which could allow storing $t$ in 
plain form using $\lg\lg\sigma$ bits per symbol). On the other hand, the sparse
method uses in general more classes, which allows for smaller sub-alphabets 
using fewer bits per symbol in sequences $s_\ell$. As we will see in 
Section~\ref{sec:exp-compr}, the dense scheme uses less space than the sparse 
one for $t$, but more for the sequences $s_\ell$.

\begin{example} \label{ex:dense}
Consider the same sequence $s = \texttt{"alabar a la alabarda"}$ of 
Ex.~\ref{ex:1}. The dense partitioning will assign $\mapping[\mathtt{a}] = 0$,
$\mapping[\mathtt{l}] = \mapping[\mathtt{'\ '}] = 1$,
$\mapping[\mathtt{b}] = \mapping[\mathtt{r}] = \mapping[\mathtt{d}] = 2$. 
So the sequence of sub-alphabet identifiers is $t[1..20] =
(\mathtt{0,1,0,2,0,2,1,0,1,1,0,1,0,1,0,2,0,2,2,0})$, and the
subsequences are $s_0 = (\mathtt{1,1,1,1,1,1,1,1,1})$, 
$s_1 = (\mathtt{2,1,1,2,1,2})$, and $s_2 = (\mathtt{1,3,1,3,2})$.

This dense scheme uses 16 bits for the sequences $s_\ell$, and the
zero-order compressed $t$ requires $n\Ho(t) = 30.79$ bits. The overall
compression is 2.34 bits per symbol. The sparse partitioning of 
Ex.~\ref{ex:1} used 10 bits in the sequences $s_\ell$, and the zero-order
compressed $t$ required $n\Ho(t) = 34.40$ bits. The total gives
2.22 bits per symbol. In our real applications, the dense partitioning
performs better.
\end{example}

Note that the question of space optimality is elusive in this scenario.
Since the encoding in $t$ plus that in the corresponding sequence $s_\ell$
forms a unique code per symbol, the optimum is reached when we choose one
sub-alphabet per symbol, so that the sequences $s_\ell$ require zero bits
and all the space is in $n\Ho(t) = n\Ho(s)$. The alphabet partitioning
always gives away some space, in exchange for faster decompression (or, in 
our case, faster $\rank$/$\select$/$\access$ operations).

Said's optimal partitioning \cite{Sai05} takes care of this problem by using
a parameter $k$ that is the maximum number of sub-alphabets to use. 
We sort the alphabet by decreasing frequency and call $S(c,k)$ the total number
of bits required to encode the symbols $[c..\sigma]$ of the alphabet using 
a partitioning into at most $k$ sub-alphabets. In general, we can make a
sub-alphabet with the symbols $[c..c']$ and solve optimally the rest, but if
$k=1$ we are forced to choose $c'=\sigma$. When we can choose, the optimization
formula is as follows: 
\[  S(c,k) = \min_{c \le c' \le \sigma} \left( f \lg\frac{n}{f} + 
				      f \lceil \lg (c'-c+1) \rceil +
				      S(c'+1,k-1) \right),
\]
where $f$ is the total frequency of symbols $c$-th to $c'$-th in $s$. The first 
term of the sum accounts for the increase in $t\Ho(s)$, the second for
the size in bits of the new sequence $s_\ell$, and the third for the smaller
subproblem, where it also holds $S(\sigma+1,k)=0$ for any $k$. This dynamic 
programming algorithm requires $\Oh{k\sigma}$ space and $\Oh{\sigma^2}$ time.
We call this partitioning method {\em optimal}.

\begin{example}
The optimal partitioning using 3 classes just like the dense approach in
Ex.~\ref{ex:dense} leaves \texttt{'a'} in its own class, then groups
\texttt{'l'}, \texttt{' '}, \texttt{'b'} and \texttt{'r'} in a second class,
and finally leaves \texttt{'d'} alone in a third class. The overall space
$n\Ho(t) + \sum |s_\ell| \lceil \lg \sigma_\ell \rceil$ is 2.23 bits per
symbol, less than the 2.34 reached by the dense partitioning. If, instead,
we let it use four classes, it gives the same solution as the sparse method in 
Ex.~\ref{ex:1}.
\end{example}

Finally, in Section~\ref{sec:partitioning} we represent the sequences 
$s_\ell$ with small alphabets $\sigma_\ell$ using wavelet trees (just 
like $t$) instead of using the representation of Golynski et al.~\cite{GMR06}, 
which is used for large $\sigma_\ell > \lg n$. In theory, this is because
Golynski et al.'s representation does not ensure sublinearity on smaller
alphabets when used inside our scheme. While this may appear to be a theoretical
issue, the implementation of such data structure (e.g., in \libcds) is indeed 
unattractive for small alphabets. For this reason, we also avoid using it on 
the chunks where $\sigma_\ell$ is small (in our case, the first ones). Note 
that using a wavelet tree for $t$ and then another for the symbols in a 
sequence $s_\ell$ is equivalent to replacing the wavelet tree leaf 
corresponding to $\ell$ in $t$ by the whole wavelet tree of $s_\ell$. The
space used by such an arrangement is worse than the one obtained by building,
from scratch, a wavelet tree for $t$ where the symbols $t[i]=\ell$ are
actually replaced by the corresponding symbol $s[i]$.

In our {\em dense} representation we use a parameter $\ell_\textrm{min}$ that
controls the minimum $\ell$ value that is represented outside of $t$. All 
the symbols that would belong to $s_\ell$, for $\ell < \ell_\textrm{min}$, are
represented directly in $t$. Note that, by default, since $\sigma_0=1$,
we have $\ell_\textrm{min}=1$.

\subsection{Compression performance}
\label{sec:exp-compr}

For all the experiments in Section~\ref{sec:exper}, except 
Section~\ref{sec:exp-graph}, we used real datasets extracted from Wikipedia. 
We considered two large collections, Simple English and Spanish, dated from
06/06/2011 and 03/02/2010, respectively. Both are regarded as sequences of
words, not characters. These collections contain several versions of each
article. Simple English, in addition, uses a reduced vocabulary. We collected 
a sample of $100{,}000$ versions at random from all the documents of Simple
English, which makes a long and repetitive sequence over a small alphabet. 
For the Spanish collection, which features a much richer vocabulary, we took 
the oldest version of each article, which yields a sequence of similar
size, but with a much larger alphabet.

We generated a single sequence containing the word identifiers of all the
articles concatenated, obtained after stemming the collections using
Porter for English and Snowball for Spanish. Table \ref{tab:data}
shows some basic characteristics of the sequences obtained%
\footnote{The code for generating these sequences is available at
\texttt{https://github.com/fclaude/txtinvlists}.}.

\begin{table}
\begin{center}
\begin{tabular}{l|r|r|r|r}
Collection & Articles & Total words $(n)$ & Distinct words $(\sigma)$
& Entropy ($\Ho(s)$) \\ \hline
{\tt Simple English} & $100{,}000$ & $766{,}968{,}140$ & $664{,}194$ & $11.60$ \\
{\tt Spanish} & $1{,}590{,}453$ & $511{,}173{,}618$ & $3{,}210{,}671$ & $11.37$ 
\end{tabular}
\caption{Main characteristics of the datasets used.}
\label{tab:data}
\end{center}
\end{table}

We measured the compression ratio achieved by the three partitioning
schemes, {\tt dense}, {\tt sparse}, and {\tt optimal}. For {\tt dense} we
did not include any individual symbols (other than the most frequent) in
sequence $t$, i.e., we let $\ell_\textrm{min}=1$. For {\tt optimal} we 
allow $1+\lfloor \lg\sigma \rfloor$ sub-alphabets, just like {\tt dense}.

In all cases, the symbols in each $s_\ell$ are represented using $\lceil \lg
\sigma_\ell \rceil$ bits. The sequence of classes $t$, instead, is represented
in three different forms: {\tt Plain} uses a fixed number of bits per symbol,
$\lceil \lg \ell \rceil$ where $\ell$ is the maximum class; {\tt Huff} uses
Huffman coding of the symbols\footnote{We use G. Navarro's Huffman
implementation; the code is available in {\sc Libcds}.}, and {\tt AC} uses
Arithmetic coding of the symbols\footnote{We use the code by J. Carpinelli,
A. Moffat, R. Neal, W. Salamonsen, L. Stuiver, A. Turpin
and I. Witten, available at 
{\tt http://ww2.cs.mu.oz.au/$\sim$alistair/arith\_coder/arith\_coder-3.tar.gz}.
We modified the decompressor to read the whole
stream before timing decompression.}. The former encodings are faster,
whereas the latter use less space. In addition we consider compressing the
original sequences using Huffman ({\tt Huffman}) and Arithmetic coding ({\tt
Arith}).

\begin{figure}[tb]
\centerline{%
\includegraphics[width=0.49\textwidth]{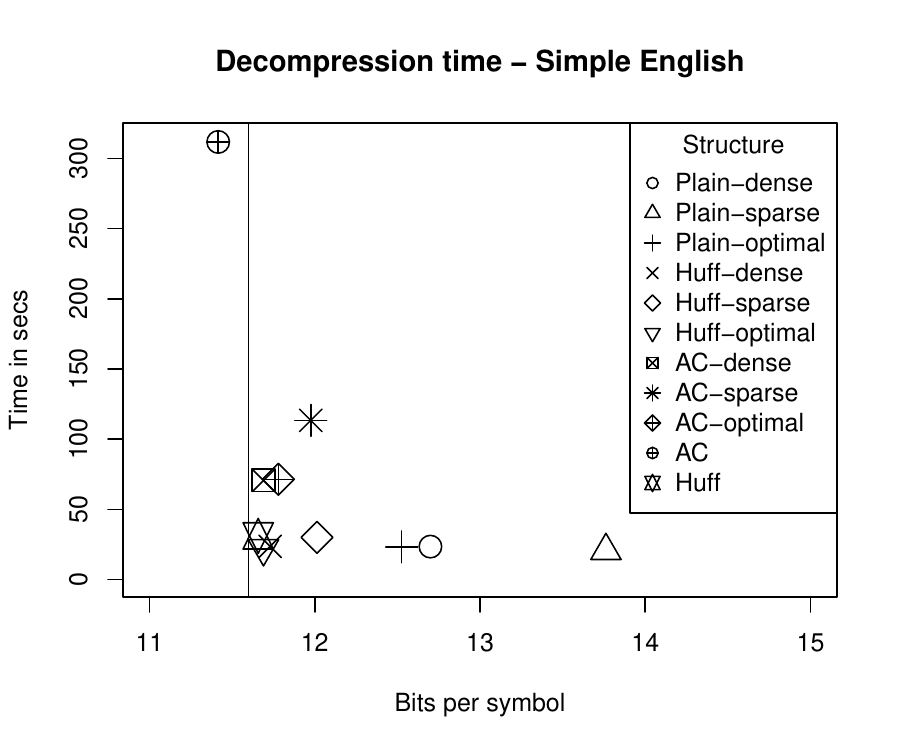}
\includegraphics[width=0.49\textwidth]{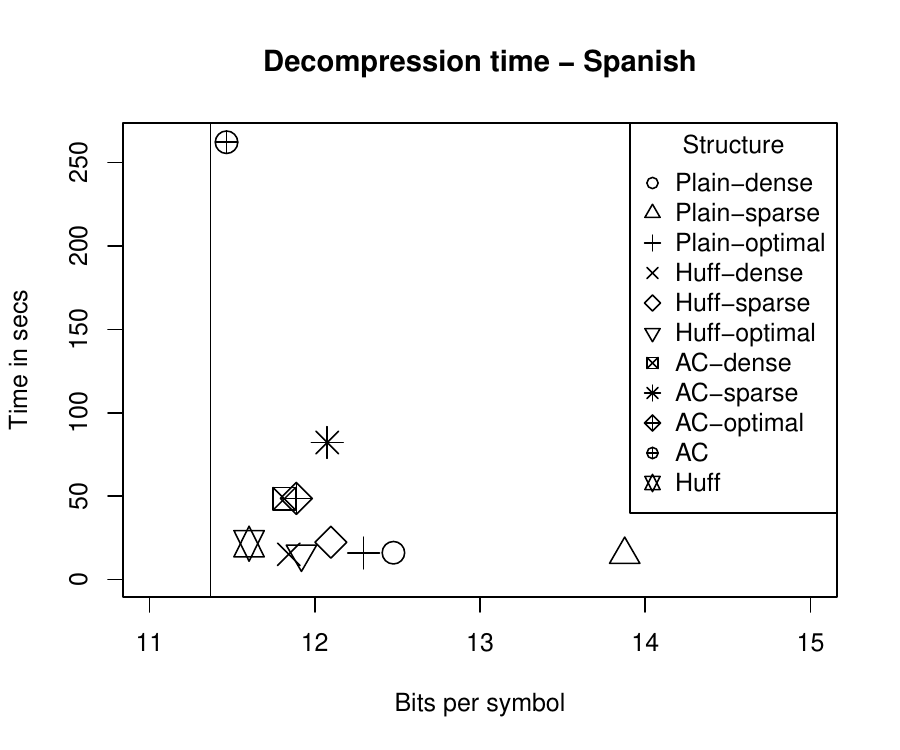}}
\caption{Space versus decompression time for basic and alphabet-partitioned
schemes. The vertical line marks the zero-order entropy of the sequences. AC
can slightly break the entropy barrier on Simple English because it is 
adaptive.}
\label{fig:spacetime}
\end{figure}

As explained, the main interest in using alphabet partitioning in a 
compressor is to speed up decompression without sacrificing too much space.
Figure~\ref{fig:spacetime} compares all these alternatives in terms of space
usage (percentage of the original sequence) and decompression time per symbol. 
It can be seen that alphabet partitioning combined with {\tt AC} compression
of $t$ wastes almost no space due to the partitioning, and speeds up 
considerably the decompression of the bare {\tt AC} compression.
However, bare {\tt Huffman} also uses the same size and decompresses several
times faster. Therefore, alphabet partitioning combined with {\tt AC}
compression is not really interesting. The other extreme is the combination
with a {\tt Plain} encoding of $t$. In the best combinations, this alphabet
partitioning wastes close to 10\% of space, and in exchange decompresses
around 30\% faster than bare {\tt Huffman}. The intermediate combination,
{\tt Huff}, wastes less than 1\% of space, while improving decompression time
by almost 25\% over bare {\tt Huffman}.

Another interesting comparison is that of partitioning methods. In all cases,
variant {\tt dense} performs better than {\tt sparse}. The difference is larger
when combined with {\tt Plain}, where {\tt sparse} is penalized for the larger
alphabet size of $t$, but still there is a small difference when combined with
{\tt AC}, which shows that $t$ has also (slightly) lower entropy in variant 
{\tt dense}. 
Table~\ref{tab:breakdown} gives a breakdown of the bits per symbol in $t$
versus the sequences $s_\ell$ in all the methods. It can be seen that 
{\tt sparse} leaves much more information on sequence $t$ than the 
alternatives, which makes it less appealing since the operation of $t$ is 
slower than that of the other sequences. However, this can be counterweighted 
by the fact that {\tt sparse} produces many more sequences with alphabet size 
1, which need no time for accessing. It is also confirmed that {\tt dense} 
leaves slightly less information on $t$ than {\tt optimal}, and that the 
difference in space between the three alternatives is almost negligible 
(unless we use {\tt Plain} to encode $t$, which is not interesting).

\begin{table}[tb]
\begin{center}
\begin{tabular}{l|rrr|rrr}
Combination & \multicolumn{3}{c|}{Simple English} & \multicolumn{3}{c}{Spanish}  \\
	    & $t$~~ & $s_\ell$~~ & $\%$~~ & $t$~~ & $s_\ell$~~ & $\%$~~ \\
\hline
{\tt Plain-dense}    & $4.96$ & $ 7.74$ & $ 6.67$ & $ 4.76$ & $ 7.72$ & $ 11.70$ \\
{\tt Plain-sparse}	& $9.97$ & $ 3.79$ & $ 19.01$ & $9.79$ & $ 4.08$ & $ 32.91$ \\
{\tt Plain-optimal}	& $4.96$ & $ 7.66$ & $ 10.21$ & $4.76$ & $ 7.61$ & $ 16.88$ \\
{\tt Huff-dense}	  & $3.99$ & $ 7.74$ & $ 6.67$ & $ 4.13$ & $ 7.72$ & $ 11.70$ \\
{\tt Huff-sparse}	  & $8.22$ & $ 3.79$ & $ 19.01$ & $8.01$ & $ 4.08$ & $ 32.91$ \\
{\tt Huff-optimal}	& $4.08$ & $ 7.66$ & $ 10.21$ & $4.24$ & $ 7.61$ & $ 16.88$ \\
{\tt AC-dense}		  & $3.95$ & $ 7.74$ & $ 6.67$ & $ 4.10$ & $ 7.72$ & $ 11.70$ \\
{\tt AC-sparse}		  & $8.18$ & $ 3.79$ & $ 19.01$ & $7.99$ & $ 4.08$ & $ 32.91$ \\
{\tt AC-optimal}	  & $4.03$ & $ 7.66$ & $ 10.21$ & $4.20$ & $ 7.61$ & $ 16.88$ \\
\end{tabular}
\caption{Breakdown, in bits per symbol, of the space used in sequence $t$
versus the space used in all the sequences $s_\ell$, for the different
combinations. The third column in each collection is the percentage of symbols
that lie in sequences $s_\ell$ with alphabet sizes $\sigma_\ell=1$.}
\label{tab:breakdown}
\end{center}
\end{table}

Finally, let us consider how the partitioning method affects decompression 
time, given an encoding method for $t$. For method {\tt AC}, {\tt sparse} is
significantly slower. This is explained by the $t$ component having many
more bits, and the decompression time being dominated by the processing of $t$
by the (very slow) arithmetic decoder. For method {\tt Plain}, instead, {\tt 
sparse} is slightly faster, despite the fact that it uses more space. Since now 
the reads on $t$ and $s_\ell$ take about the same time, this difference is 
attributable to the fact that {\tt sparse} leaves more symbols on sequences 
$s_\ell$ with alphabets of size 1, where only one read in $t$ is needed to 
decode the symbol (see Table~\ref{tab:breakdown}).
For {\tt Huff} all the times are very similar, and very close to the fastest
one. Therefore, for the rest of the experiments we use the variant {\tt Huff}
with {\tt dense} partitioning, which performs best in space/time.

\subsection{Rank, select and access}
\label{sec:exp-seqs}

We now consider the efficiency in the support for the operations \rank,
\select, and \access. We compare our sequence representation with the
state of the art, as implemented in \libcds\ v1.0.10, a library of
highly optimized implementations of compact data structures. As said,
\libcds\ already includes the implementation of our new structure.

We compare six data structures for representing sequences. Those based on 
wavelet trees are obtained in \libcds\ by combining sequence representations 
(\verb|WaveletTreeNoptrs|, \verb|WaveletTree|) with bitmap representations
(\verb|BitSequenceRG|, \verb|BitSequenceRRR|) for the data on wavelet tree 
nodes.  

\begin{itemize}
\item \verb|WTNPRG|: Wavelet tree without pointers, obtained as
  \verb|WaveletTreeNoptrs|$+$\verb|BitSequenceRG| in \libcds. 
  This corresponds to the 
  basic balanced wavelet tree structure \cite{GGV03}, where all the bitmaps 
  of a level are concatenated \cite{MN07}. The bitmaps are represented in 
  plain form and their operations are implemented using a one-level directory 
  \cite{GGMN05} (where \rank\ is implemented in time proportional to a
  sampling step and \select\ uses a binary search on \rank). The space is
  $n\lg\sigma + o(n\lg\sigma)$ and the times are $\Oh{\lg\sigma}$. In
  practice the absence of pointers yields a larger number of operations to
  navigate in the wavelet tree, and also \select\ operation on bitmaps is
  much costlier than \rank. A space/time tradeoff is obtained by varying
  the sampling step of the bitmap \rank\ directories.
\item \verb|WTNPRRR|: Wavelet tree without pointers with bitmap compression, 
  obtained in \libcds\ as \verb|WaveletTreeNoptrs|$+$\verb|BitSequenceRRR|. 
  This is similar to \verb|WTNPRG|, but the bitmaps are represented in
  compressed form using the FID technique \cite{RRR02} (\select\ is also
  implemented with binary search on \rank). The space is
  $n\Ho(s) + o(n\lg\sigma)$ and the times are $\Oh{\lg\sigma}$. In practice
  the FID representation makes it considerably slower than the version with
  plain bitmaps, yet \select\ operation is less affected.
  A space/time tradeoff is obtained by varying the sampling
  step of the bitmap \rank\ directories.
\item \verb|GMR|: The representation proposed by Golysnki et al.~\cite{GMR06}, 
  named \verb|SequenceGMR| in \libcds. The space is $n\lg\sigma+o(n\lg\sigma)$,
  yet the lower-order term is sublinear on $\sigma$, not $n$. The time is
  $\Oh{1}$ for \select\ and $\Oh{\lg\lg\sigma}$ for \rank\ and \access, although
  on average \rank\ is constant-time. A space/time tradeoff, which in practice
  affects only the time for \access, is obtained by varying the permutation 
  sampling inside the chunks \cite{GMR06}.
\item \verb|WTRG|: Wavelet tree with pointers and Huffman shape, obtained as
  \verb|WaveletTree|$+$\verb|BitSequenceRG| in \libcds. 
  The space is $n\Ho(s) + \Oh{n} + o(n\Ho(s)) + \Oh{\sigma\lg n}$. The time is 
  $\Oh{\lg\sigma}$, but in our experiments it will be $\Oh{\Ho(s)}$ for \access,
  since the positions are chosen at random from the sequence and then we 
  navigate less frequently to deeper Huffman leaves.
\item \verb|WTRRR|: Wavelet tree with pointers, obtained with
  \verb|WaveletTree|+ \verb|BitSequenceRRR| in \libcds. The space is
  $n\Ho(s) + o(n\Ho(s)) + \Oh{\sigma\lg n}$. The time is as in the previous
  structure, except that in practice the FID representation is considerably
  slower.
\item \verb|AP|: Our new alphabet partitioned structure, named 
  \verb|SequenceAlphPart| in \libcds. We use dense partitioning and include
  the $2^{10}$ most frequent symbols directly in $t$, $\ell_\textrm{min}=10$. 
  Sequence $t$ is represented with a \verb|WTRG| (since its alphabet is small 
  and the pointers pose no significant overhead), and the sequences 
  $\sigma_\ell$ are represented with structures \verb|GMR|. The space is 
  $n\Ho(s)+o(n\Ho(s))$, although the lower-order term is actually sublinear on 
  $\sigma$ (and only
  very slightly on $n$). The times are as in \verb|GMR|, although there is a
  small additive overhead due to the wavelet tree on $t$. A space/time tradeoff
  is obtained with the permutations sampling, just as in \verb|GMR|.
\end{itemize}

Figure \ref{fig:opers} shows the results obtained for both text collections,
giving the average over $100{,}000$ measures. The \rank\ queries were generated
by choosing a symbol from $[1..\sigma]$ and a position from $[1..n]$, both
uniformly at random. For \select\ we chose the symbol $a$ in the same way, and
the other argument uniformly at random in $[1..|s|_a]$. Finally, for \access\ 
we generated the position uniformly at random in $[1..n]$. Note that the latter
choice favors Huffman-shaped wavelet trees, on which we descend to leaf $a$ with
probability $|s|_a/n$, whereas for \rank\ and \select\ we descend to any leaf
with the same probability.

\begin{figure}
\centerline{%
\includegraphics[height=0.3\textheight,width=0.49\textwidth]{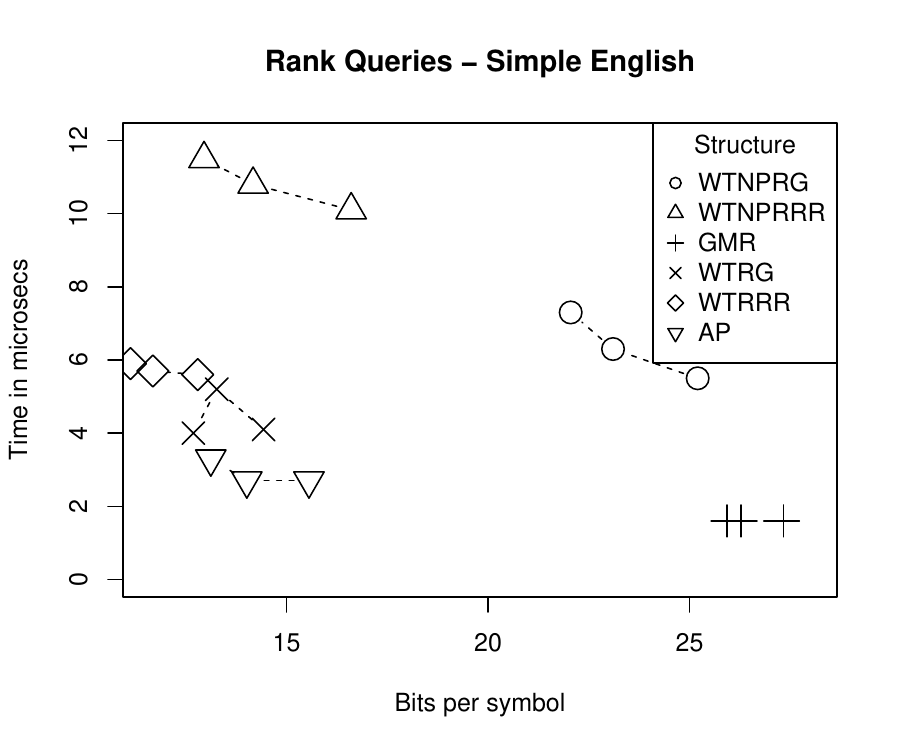}
\includegraphics[height=0.3\textheight,width=0.49\textwidth]{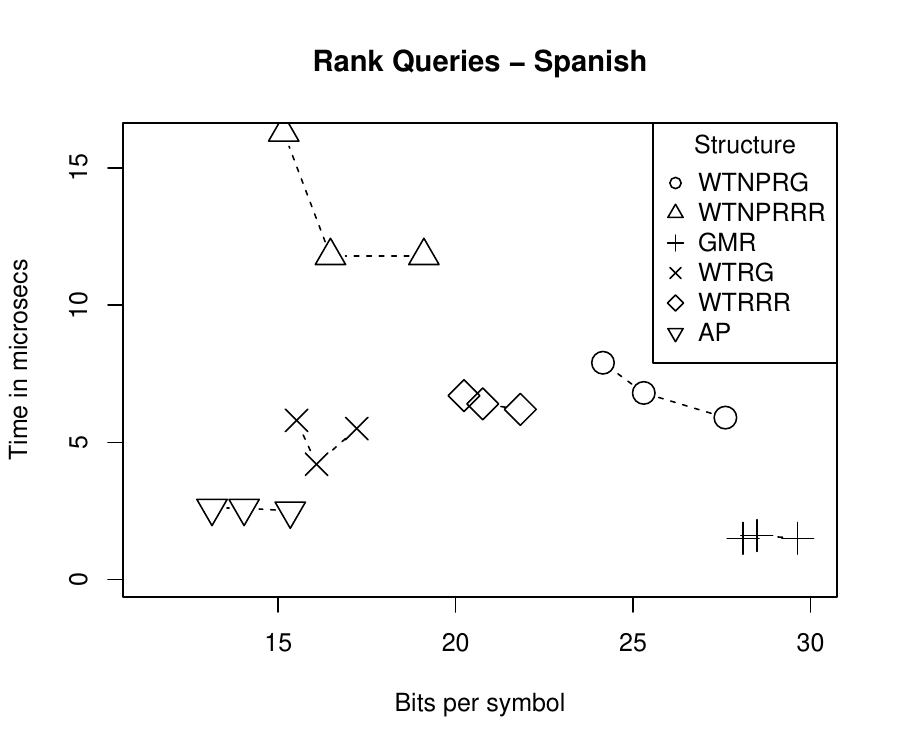}}

\centerline{%
\includegraphics[height=0.3\textheight,width=0.49\textwidth]{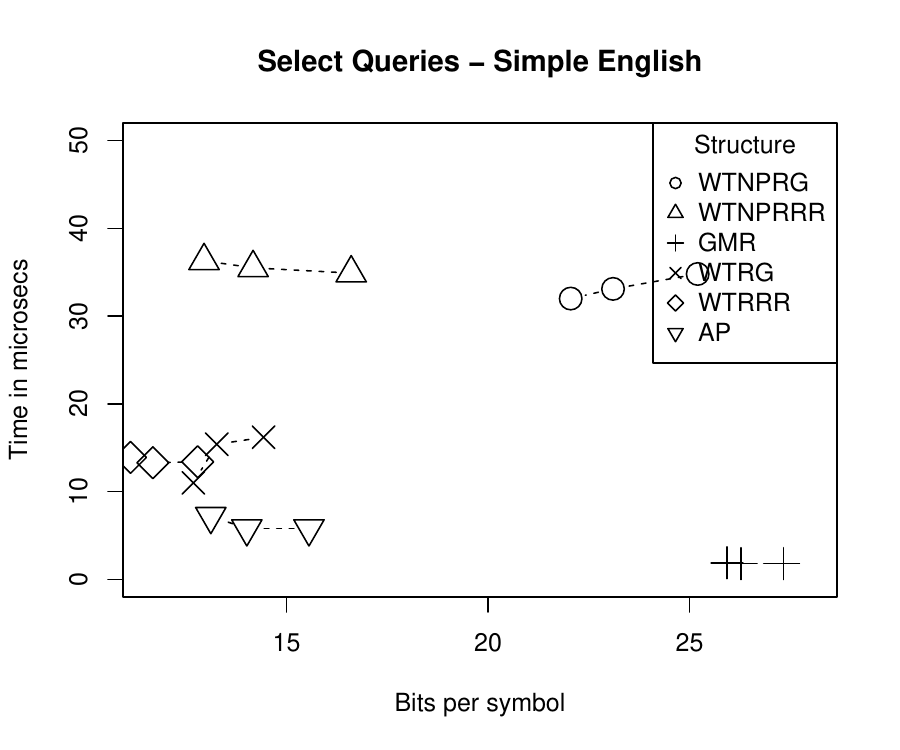}
\includegraphics[height=0.3\textheight,width=0.49\textwidth]{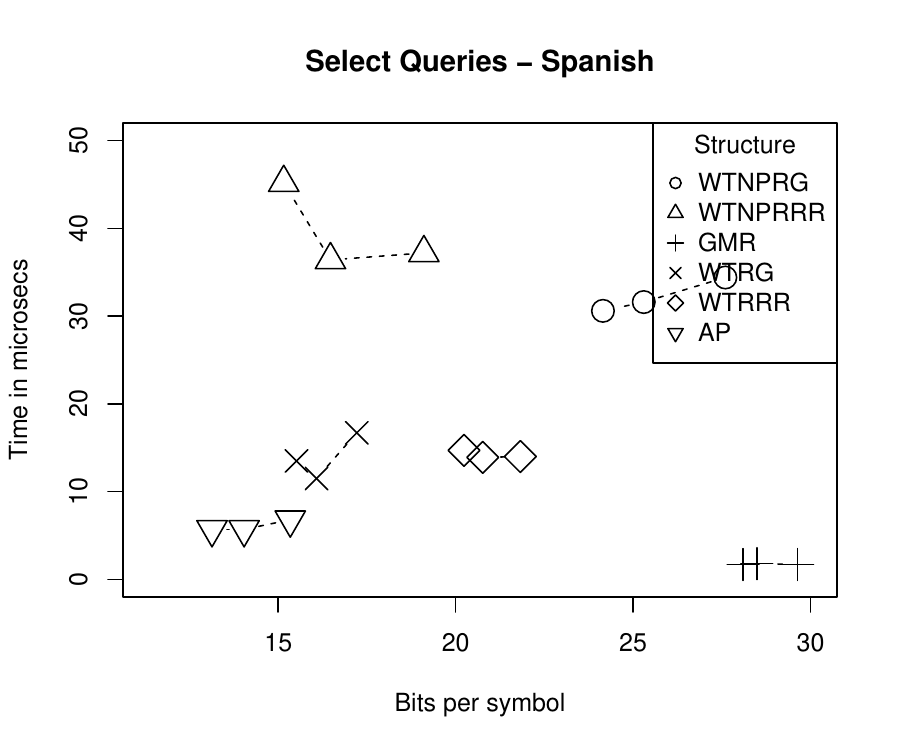}}

\centerline{%
\includegraphics[height=0.3\textheight,width=0.49\textwidth]{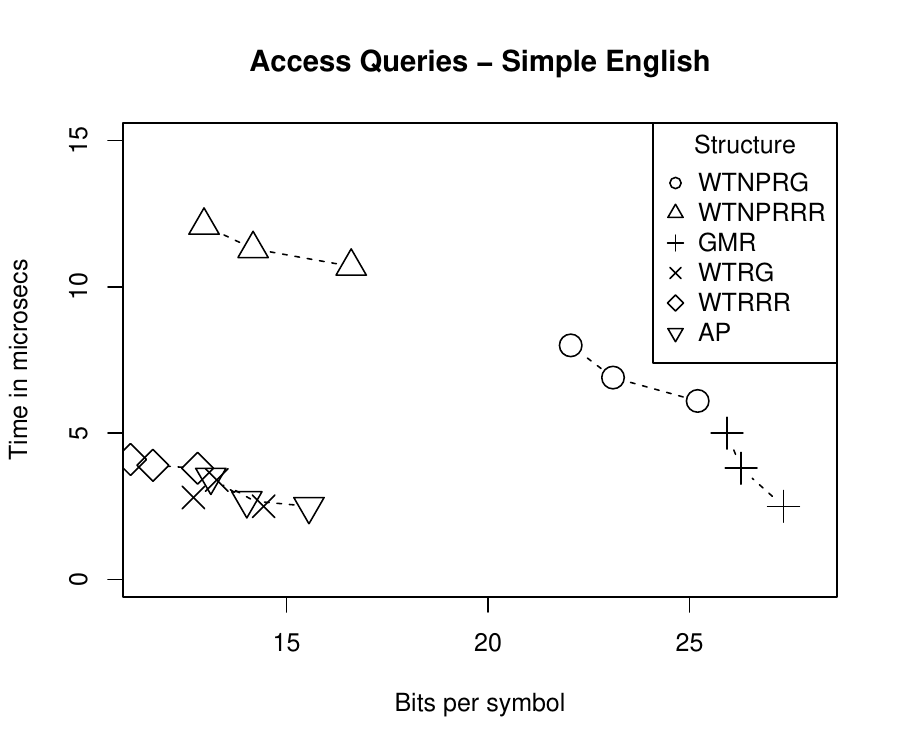}
\includegraphics[height=0.3\textheight,width=0.49\textwidth]{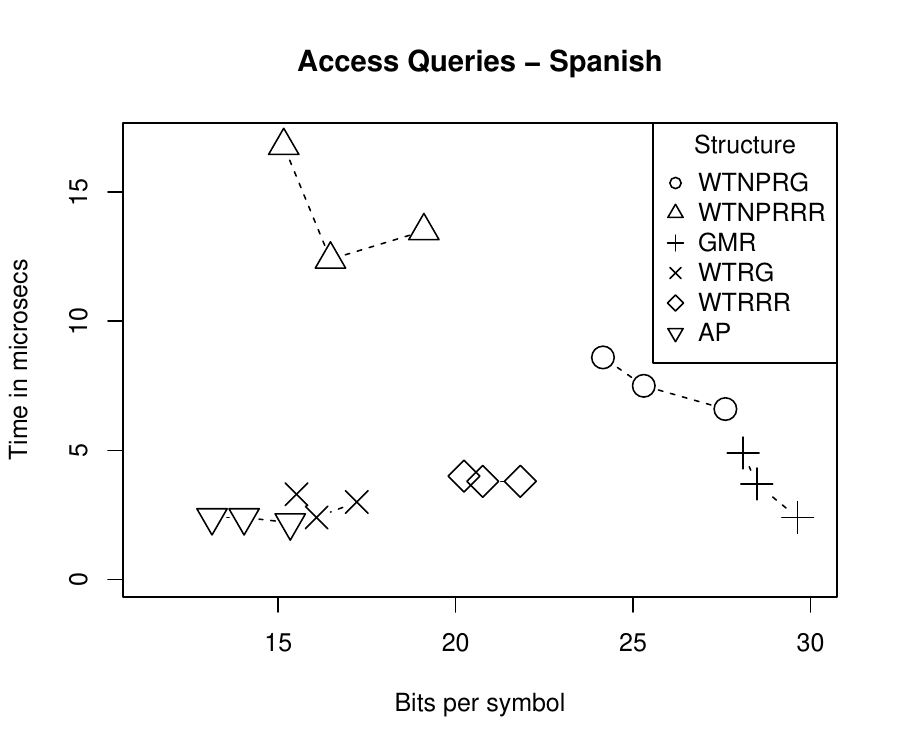}}

\caption{Time for the three operations. The $x$ axis starts at the entropy of the sequence.}
\label{fig:opers}
\end{figure}

Let us first analyze the case of Simple English, where the alphabet is
smaller. Since $\sigma$ is 1000 times smaller than $n$, the 
$\Oh{\sigma\lg n}$ terms of Huffman-shaped wavelet trees are not significant,
and as a result the variant \texttt{WTRRR} reaches the least space,
essentially $n\Ho(s)+o(n\Ho(s))$. It is followed by three variants that use
similar space: \texttt{WTRG} (which has an additional $\Oh{n}$-bit overhead),
\texttt{AP} (whose $o(n\Ho(s))$ space term is higher than that of wavelet
trees), and \texttt{WTNPRRR} (whose sublinear space term is of the form
$o(n\lg\sigma)$, that is, uncompressed). The remaining structures,
\texttt{WTNPRG} and \texttt{GMR}, are not compressed and use much more space.

In terms of time, structure \texttt{AP} is faster than all the others
except \texttt{GMR} (which in exchange uses much more space). The exception
is on \access\ queries, where as explained Huffman-shaped wavelet trees,
\texttt{WTRG} and \texttt{WTRRR}, are favored and reach the same performance
of \texttt{AP}. In general, the rule is that variants using plain bitmaps are
faster than those using FID compression, and that variants using pointers
and Huffman shape are faster than those without pointers (as the latter need
additional operations to navigate the tree). These differences are smaller on
\select\ queries, where the binary searches dominate most of the time spent.

The Spanish collection has a much larger alphabet: $\sigma$ is only 100 times
smaller than $n$. This impacts on the $\Oh{\sigma\lg n}$ bits used by the
pointer-based wavelet trees, and as a result the space of \texttt{AP},
$n\Ho(s) + o(n\Ho(s))$, is unparalleled. Variants \texttt{WTRRR} and
\texttt{WTRG} use significantly more space and are followed, far away,
by \texttt{WTNPRRR}, which has uncompressed redundancy. The
uncompressed variants \texttt{WTNPRG} and \texttt{RG} use significantly more
space. The times are basically as on Simple English.

This second collection illustrates more clearly that, for large alphabets,
our structure \texttt{AP} sharply dominates the whole space/time tradeoff.
It is only slightly slower than \texttt{GMR} in some cases, but in exchange
it uses half the space. From the wavelet trees, the most competitive alternative
is \texttt{WTRG}, but it always loses to \texttt{AP}. The situation is not too
different on smaller alphabets (as in Simple English), except that variant
\texttt{WTRRR} uses clearly less space, yet at the expense of doubling
the operation times of \texttt{AP}.

\subsection{Intersecting inverted lists}
\label{sec:exp-invl}

An interesting application of \rank/\select\ operations on large alphabets
was proposed by Clarke et al.~\cite{CCT00}, and recently implemented by
Arroyuelo et al.~\cite{AGO10} using wavelet trees. The idea is to represent
the text collections as a sequence of word tokens (as done for Simple English
and Spanish), use a compressed and \rank/\select/\access-capable sequence 
representation for them, and use those operations to emulate an inverted index 
on the collection, without spending any extra space on storing explicit 
inverted lists. 

More precisely, given a collection of $d$ documents $T_1,T_2,\ldots,T_d$, we
concatenate them in $\mathcal{C} = T_1 T_2 \ldots T_d$, and build
an auxiliary bitmap $b[1..|\mathcal{C}|]$ where we mark the
beginning of each document with a 1. We can provide access to the text of any
document in the collection via \access\ operations on sequence
$\mathcal{C}$ (and \select\ on $b$). In order to emulate the inverted list of 
a given term $w$, we just need to list all the distinct documents where $w$ 
occurs. This is achieved by iterating on procedure 
{\em nextDoc}$(\Cs,w,p)$ of Algorithm~\ref{alg:nextdoc} (called initially with 
$p=0$ and then using the last $p$ value returned).

\begin{algorithm}[htb]
\SetKwInOut{Input}{input}\SetKwInOut{Output}{output}
%\SetAlgoLined
\Input{$\mathcal{C}, w, p$}
\Output{next document after $T_p$ that contains $w$}
\SetKwFunction{nextDoc}{nextDoc}

$pos \leftarrow b.\select_1(p+1)$\;
$cnt \leftarrow \Cs.\rank_w(pos-1)$\;
\Return{$b.\rank_1(\Cs.\select_w(cnt+1))$}
\caption{Function {\em nextDoc}$(\Cs, w, p)$, retrieves the next document 
after $p$ containing $w$. The $x$ axis starts at the entropy of the sequence.}
\label{alg:nextdoc}
\end{algorithm}

Algorithm \ref{alg:nextdoc} also allows one to test whether a given document 
contains a term or not ($p$ contains $w$ iff $p=\mathit{nextDoc}(\Cs,w,p-1)$). 
Using this primitive we implemented Algorithm~\ref{alg:intersect}, which
intersects several lists (i.e., returns the documents where all the given
terms appear) based on the algorithm by Demaine et al.~\cite{DLM00}. We
tested this algorithm for both Simple English and Spanish collections,
searching for phrases extracted at random from the collection. We considered 
phrases of lengths 2 to 16. We averaged the results over $1{,}000$ queries. As all 
the results were quite similar, we only show the cases of 2 and 6 words.
We tested the same structures as in Section~\ref{sec:exp-seqs}.

\begin{algorithm}[htb]
\SetKwInOut{Input}{input}\SetKwInOut{Output}{output}
%\SetAlgoLined
\SetKwFunction{nextDoc}{\textit{nextDoc}}
\Input{$\mathcal{C}, W=w_1,w_2,\ldots,w_k$}
\Output{documents that contain $w_1,\ldots,w_k$}

sort $W$ by increasing number of occurrences in the collection\;
$res \leftarrow \emptyset$\;
$p \leftarrow \nextDoc(\mathcal{C},w_1,0)$\;
\While{$p$ is valid}{
  \If{$w_2,\ldots,w_k$ are contained in $p$ (i.e., $p=\nextDoc(\Cs,w_j,p-1)$
			for $2 \le j \le k$)}{
    Add $p$ to $res$\;
    $p \leftarrow \nextDoc(\mathcal{C},w_1,p)$
    }
  \Else {
    Let $w_j$ be the first word not contained in $p$\;
    $p \leftarrow \nextDoc(\mathcal{C},w_1,\nextDoc(\mathcal{C},w_j,p-1))$
  }  
}
\Return{$res$}
\caption{Retrieving the documents where all $w_1,\ldots,w_k$ appear.}
\label{alg:intersect}
\end{algorithm}

Figure \ref{fig:inters} shows the results obtained by the different 
structures. For space, of course, the results are as before: \texttt{AP}
is the best on Spanish and is outperformed by \texttt{WTRRR} on Simple
English. With respect to time, we observe that Huffman-shaped wavelet trees
are favored compared to the random \rank\ and \select\ queries of
Section~\ref{sec:exp-seqs}. The reason is that the queries in this 
application, at least in the way we have generated them, do not distribute
uniformly at random: the symbols for \rank\ and \select\ are chosen
according to their probability in the text, which favors Huffman-shaped
trees. As a result, structures \texttt{WTRG} perform similarly to 
\texttt{AP} in time, whereas \texttt{WTRRR} is less than twice as slow.

\begin{figure}[tb]
\centerline{%
\includegraphics[width=0.49\textwidth]{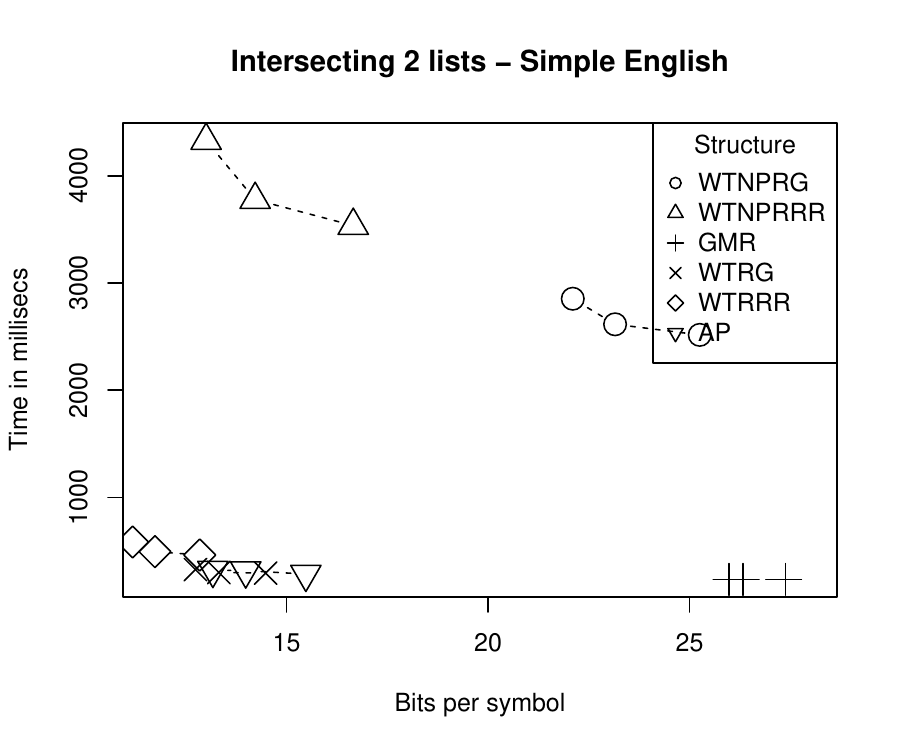}
\includegraphics[width=0.49\textwidth]{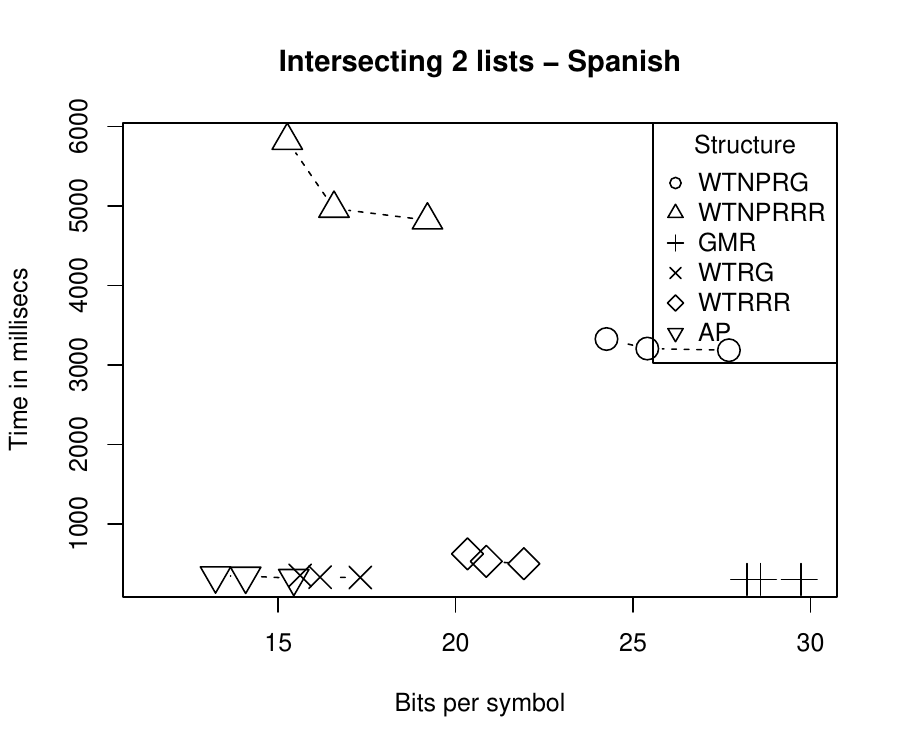}}

\centerline{%
\includegraphics[width=0.49\textwidth]{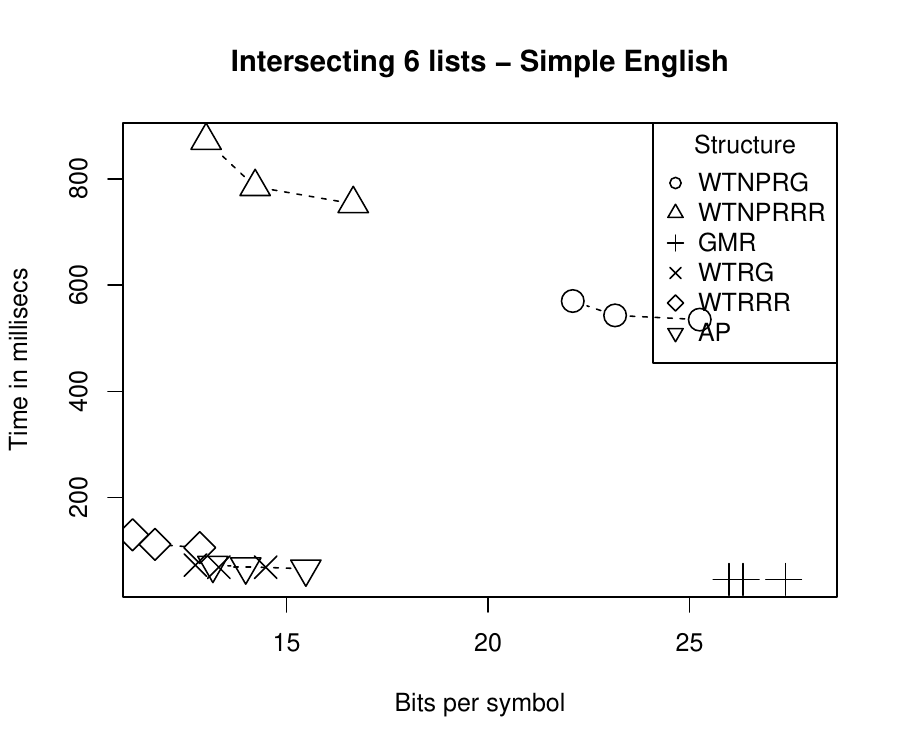}
\includegraphics[width=0.49\textwidth]{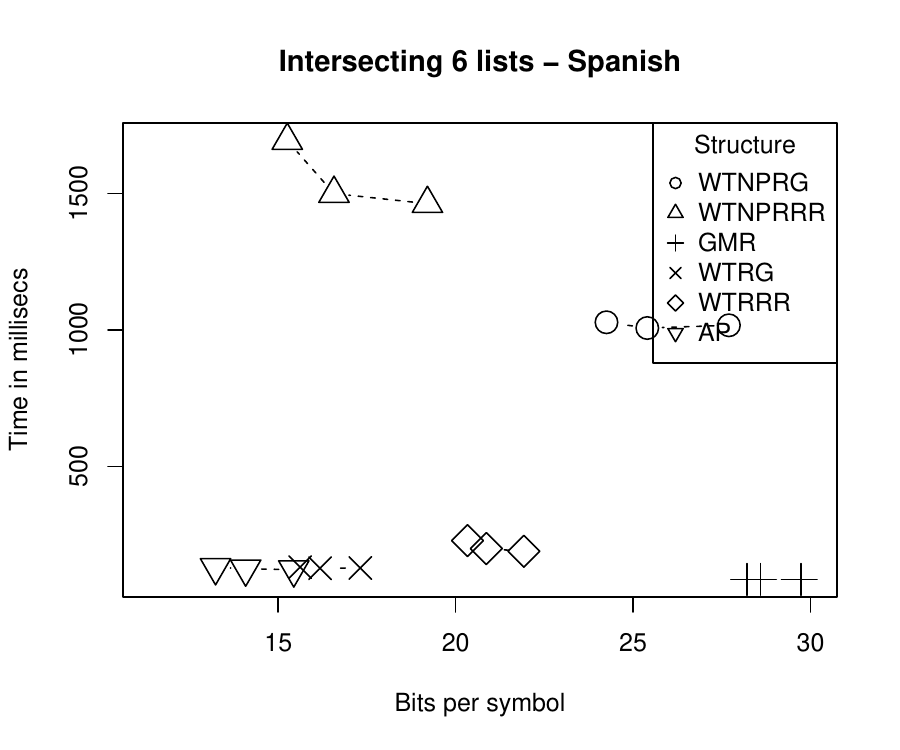}}
\caption{Results for intersection queries. The $x$ axis starts at the entropy of the sequence.}
\label{fig:inters}
\end{figure}

\subsection{Self-indexes}
\label{sec:exp-ssa}

A second application of the sequence operations on large alphabets was
explored by Fari\~na et al.~\cite{FBNCPR11}. The idea is to take 
a self-index \cite{NM07} designed for text composed of characters, and apply 
it to a word-tokenized text, in order to carry out word-level searches on
natural language texts. This requires less space and time than the
character-based indexes and competes successfully with word-addressing
inverted indexes. One of the variants they explore is to build an FM-index 
\cite{FM05,FMMN07} on words \cite{CN08}. The FM-index represents the
Burrows-Wheeler transform (BWT) \cite{BW94} $s^\textrm{bwt}$ of $s$.
Using \rank\ and \access\ operations on $s^\textrm{bwt}$ the FM-index can, 
among other operations, {\em count} the number of occurrences of a pattern 
$p[1..k]$ (in our case, a phrase of $k$ words) in $s[1..n]$. This requires
$O(k)$ applications of \rank\ and \access\ on $s^\textrm{bwt}$. A self-index 
is also able to retrieve any passage of the original sequence $s$.

We implemented the word-based FM-index with the same structures measured so 
far, plus a new variant called \verb|APRRR|. This is a version of \verb|AP| 
where the bitmaps of the wavelet tree of $t$ are represented using FIDs
\cite{RRR02}. The reason is that it was proved \cite{MN07impl} that the
wavelet tree of $s^\textrm{bwt}$, if the bitmaps are represented using
Raman et al.'s FID \cite{RRR02}, achieves space $n\Hk(s)+o(n\lg\sigma)$.
Since the wavelet tree $t$ of sub-alphabets of $s^\textrm{bwt}$ is a 
coarsened version of that of $s^\textrm{bwt}$, we expect it to take advantage
of Raman et al's representation.

We extracted phrases at random text positions, of lengths 2 to 16, and 
counted their number of occurrences using the FM-index. We averaged the
results over $100{,}000$ searches. As the results
are similar for all lengths, we show the results for lengths 2 and 8.
Figure \ref{fig:count} shows the time/space tradeoff obtained. 

\begin{figure}[tb]
\centerline{%
\includegraphics[width=0.45\textwidth]{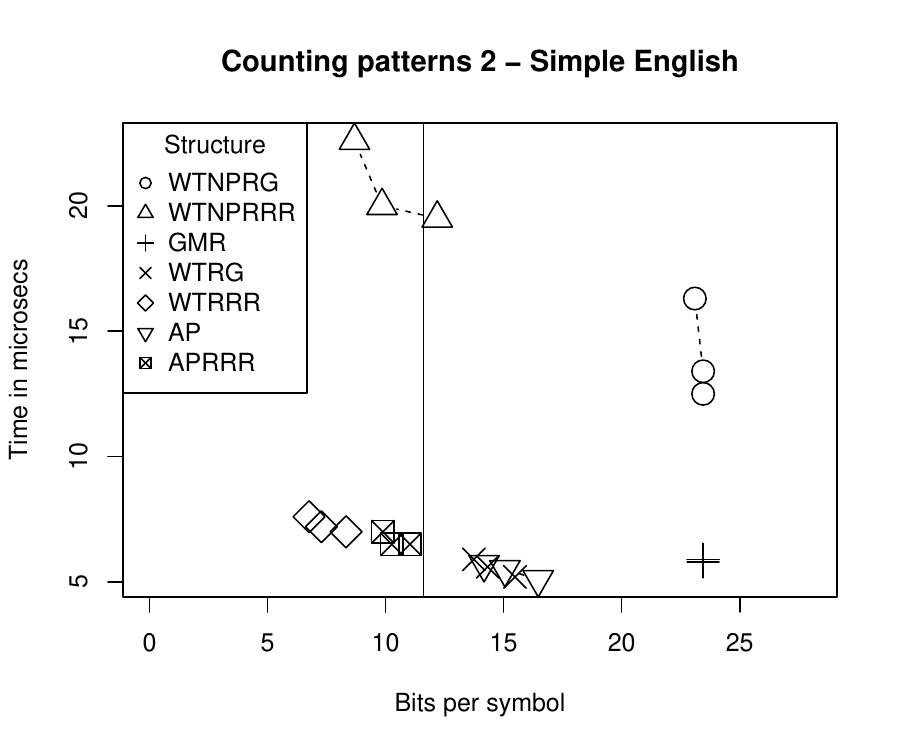}
\includegraphics[width=0.45\textwidth]{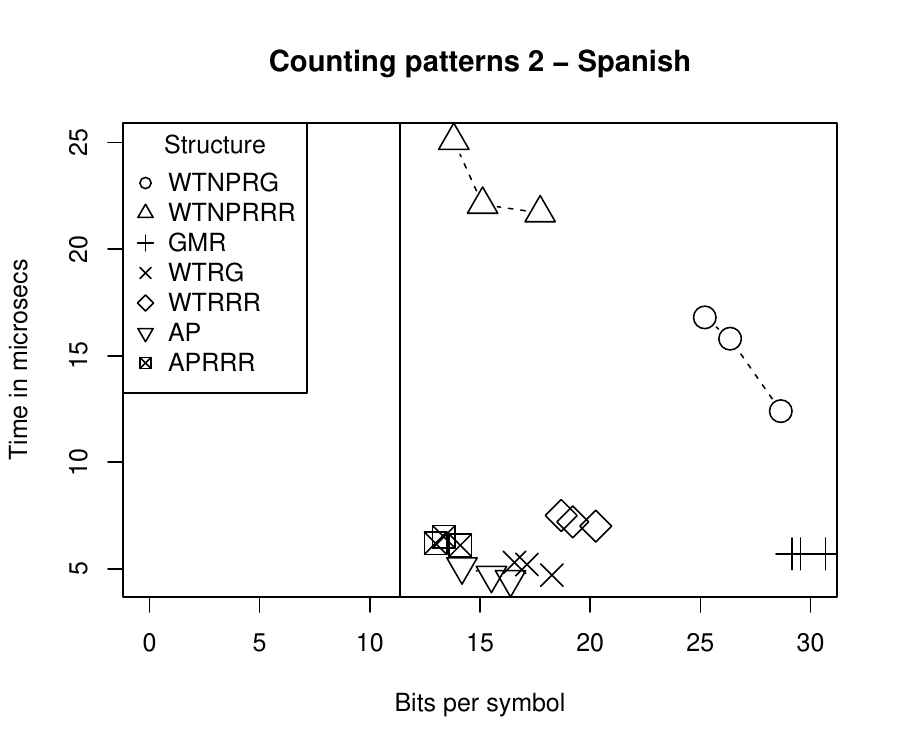}}

\centerline{%
\includegraphics[width=0.45\textwidth]{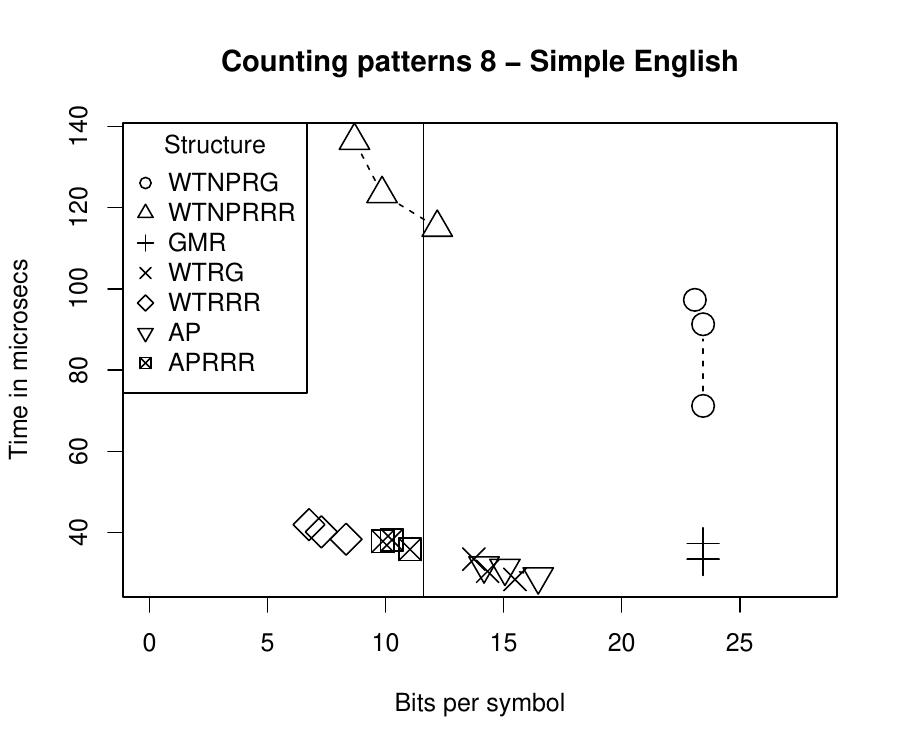}
\includegraphics[width=0.45\textwidth]{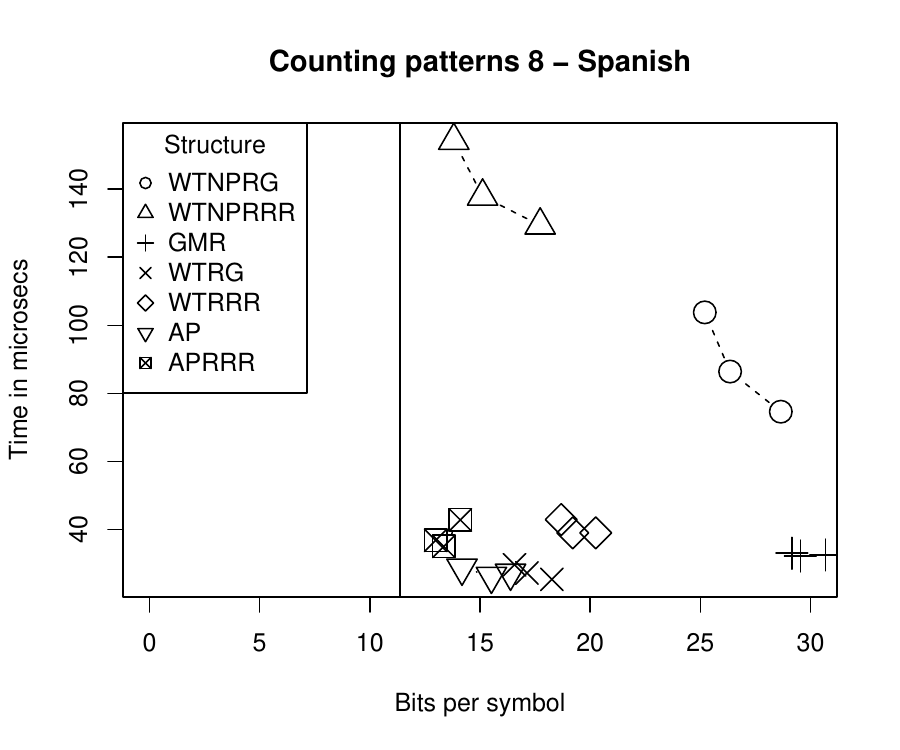}}
\caption{Time for counting queries on word-based FM-indexes. The vertical line marks the zero-order entropy of the sequences; remember that some schemes achieve high-order entropy spaces.}
\label{fig:count}
\end{figure}

Confirming the theoretical results \cite{MN07impl}, the versions using
compressed bitmaps require much less space than the other alternatives.
In particular, \texttt{APRRR} uses much less space than \texttt{AP}, 
especially on Simple English. In this text the least space is reached by
\texttt{WTRRR}. On Spanish, instead, the $\Oh{\sigma\lg n}$ bits of 
Huffman-shaped wavelet trees become
relevant and the least space is achieved by \texttt{APRRR}, closely followed
by \texttt{AP} and \texttt{WTNPRRR}. The space/time tradeoff is dominated
by \texttt{APRRR} and \texttt{AP}, the two variants of our structure.

\subsection{Navigating graphs}
\label{sec:exp-graph}

Finally, our last application scenario is the compact representation of
graphs. Let $G=(V,E)$ be a directed graph. If we concatenate the adjacency
lists of the nodes, the result is a sequence $s[1..|E|]$ over an alphabet of
size $|V|$. If we add a bitmap $b[1..|E|]$ that marks with a 1 the beginning
of the lists, it is very easy to retrieve the adjacency list of any node $v
\in V$, that is, its {\em neighbors}, with one \select\ operation  on $b$
followed by one \access\ operation on $s$ per neighbor
retrieved.\footnote{Note that this works well as long as each node points to
at least one node. We solve this problem by keeping an additional bitmap
marking the nodes whose list is not empty.}

It is not hard to reach this space with a classical graph representation.
However, classical representations do not allow one to retrieve efficiently the
{\em reverse neighbors} of $v$, that is, the nodes that point to it. The
classical solution is to double the space to represent the transposed graph.
Our sequence representation, however, allows us to retrieve the reverse
neighbors using \select\ operations on $s$, much as
Algorithm~\ref{alg:nextdoc} retrieves the documents where a term $w$ appears:
our ``documents'' are the adjacency lists of the nodes, and the document
identifier is the node $v \in V$ that points to the desired node. Similarly,
it is possible to determine whether a given node $v$ points to a given node
$v'$, which is not an easy operation with classical adjacency lists. This
idea has not only been used in this simple form \cite{CN08}, but also in 
more sophisticated scenarios where it was combined with grammar compression
of the adjacency lists, or with other transformations, to compress Web graphs 
and social networks \cite{CNtweb10,CNlncs10,HN11}.

For this experiment we used two crawls obtained from the well-known 
{\em WebGraph} project\footnote{{\tt http://law.dsi.unimi.it}}.
The main characteristics of these crawls are shown in Table \ref{tab:crawls}. 
Note that the alphabets are comparatively much larger than on documents,
just around 22--26 times smaller than the sequence length.

Figure \ref{fig:graphs} shows the results obtained. The nodes are sorted
alphabetically by URL. A well-known property of Web graphs \cite{BV04} is that 
nodes tend to point to other nodes of the same domain. This property turns into
substrings of nearby symbols in the sequence, and this turns into runs of
0s or 1s in the bitmaps of the wavelet trees. This makes variants like
\texttt{WTNPRRR} very competitive in space, whereas \texttt{APRRR} does not
benefit so much. The reason is that the partitioning into
classes reorders the symbols, and the property is lost. Note that variant
{\tt WTRRR} does not perform well in space, since the number of nodes is too
large for a pointer-based tree to be advantageous. For the same reason, even
{\tt WTRG} uses more space than {\tt GMR}\footnote{Note that \texttt{WTRRR} is
almost 50\% larger than \texttt{WTRG}. This is because the former is a more
complex structure and requires a larger (constant) number of pointers to be
represented. Multiplying by the $\sigma$ nodes of the Huffman-shaped wavelet
tree makes a significant difference when the alphabet is so large.}.
Overall, we note that our variants largely dominate the space/time tradeoff,
except that \texttt{WTNPRRR} uses less space (but much more time).

\begin{table}
\begin{center}
\begin{tabular}{l|c|c|l}
{ Name} & {Nodes } & {Edges } & ~Plain adj. list (bits per edge)\\
\hline
EU (EU-2005) & $862{,}664$ & $19{,}235{,}140$ & $~20.81$ \\
In (Indochina-2002) & $7{,}414{,}866$ & $194{,}109{,}311$ & $~23.73$ \\
\end{tabular}
\caption{Description of the Web crawls considered.}
\label{tab:crawls}
\end{center}
\end{table}

\begin{figure}[tb]
\centerline{%
\includegraphics[width=0.49\textwidth]{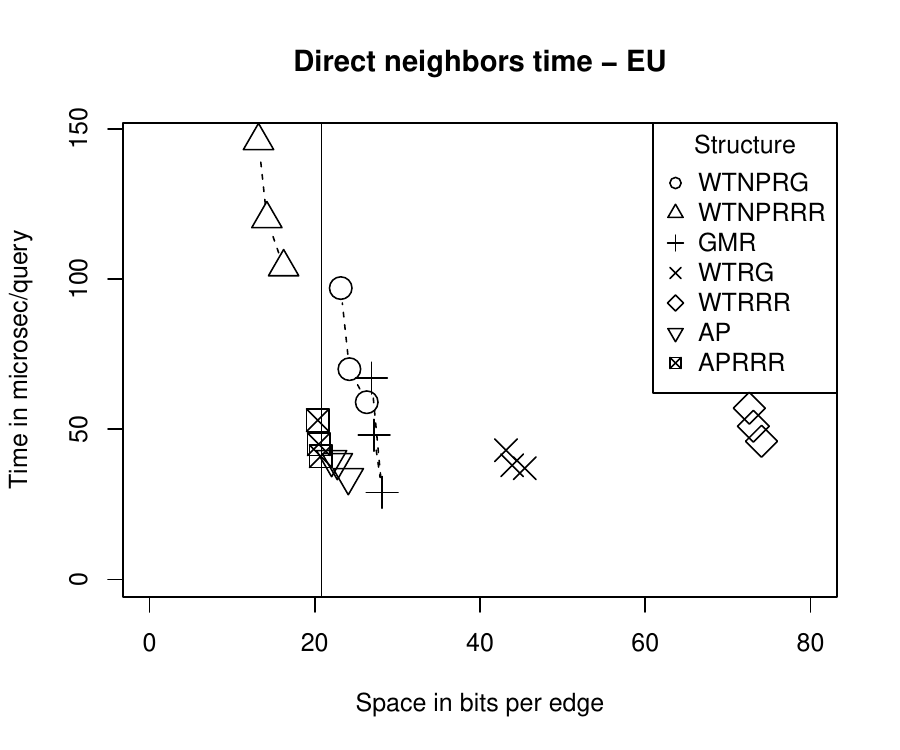}
\includegraphics[width=0.49\textwidth]{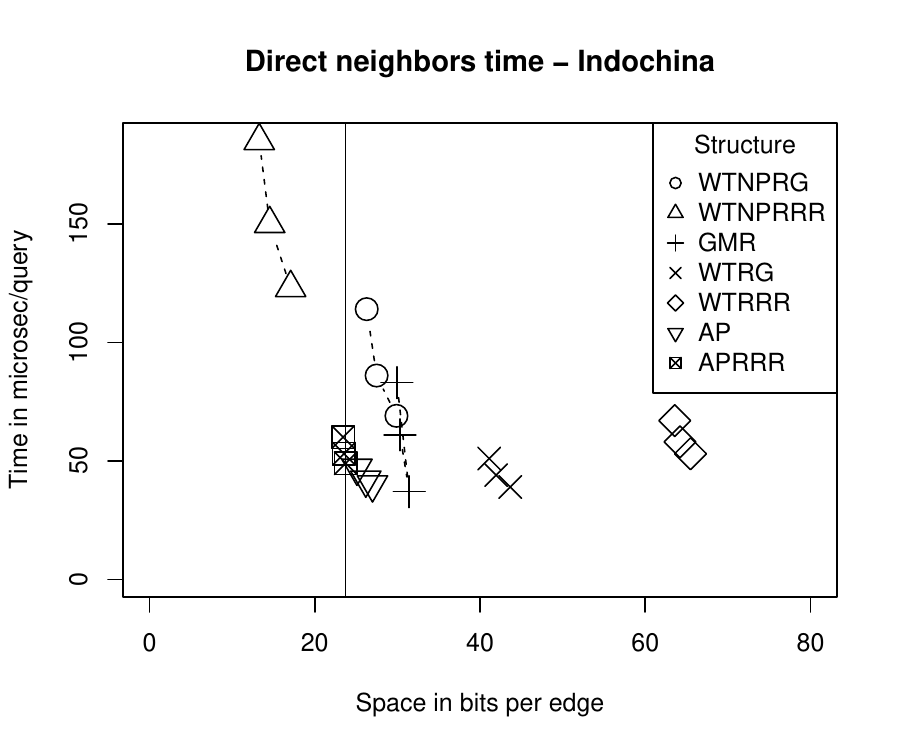}}

\centerline{%
\includegraphics[width=0.49\textwidth]{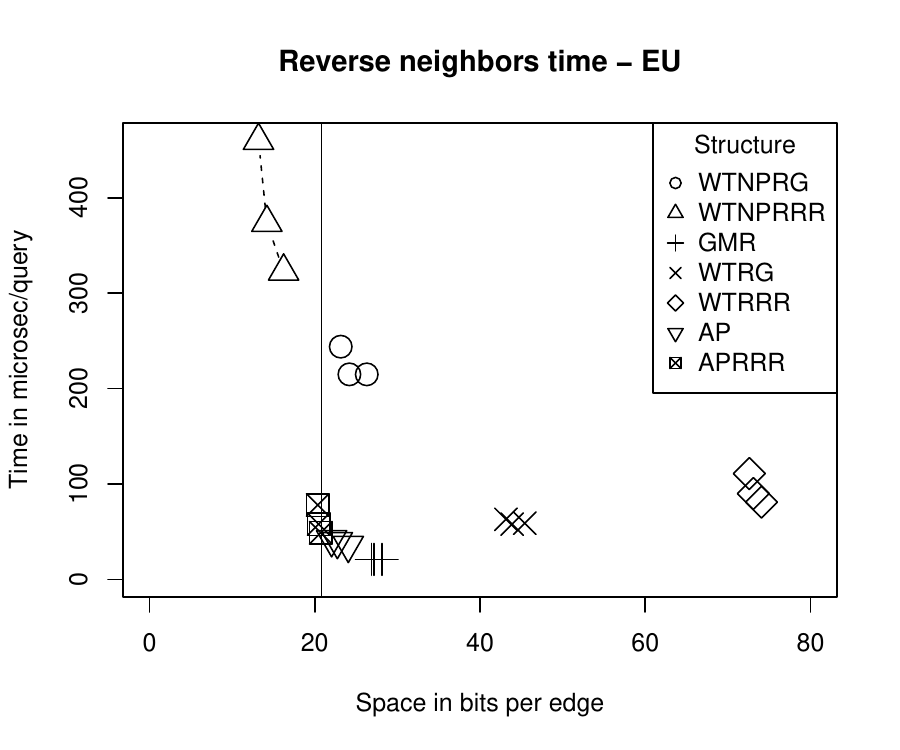}
\includegraphics[width=0.49\textwidth]{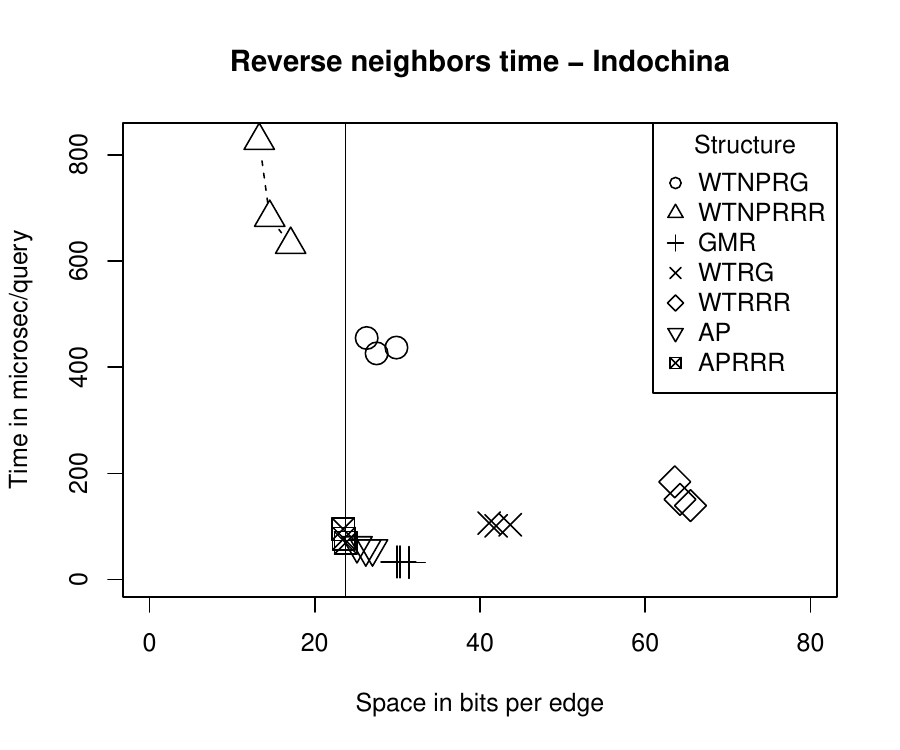}}
\caption{Performance on Web graphs, to retrieve direct and reverse neighbors. The vertical line marks the bits per edge required by a plain adjacency list representation.}
\label{fig:graphs}
\end{figure}

\section{Conclusions and future work}

We have presented the first zero-order compressed representation of sequences
supporting queries \access, \rank, and \select\ in loglogarithmic time, so that 
the
redundancy of the compressed representation is also compressed. That is, our
space for sequence $s[1..n]$ over alphabet $[1..\sigma]$ is $n\Ho(s) +
o(n)(\Ho(s)+1)$ instead of the usual $n\Ho(s)+o(n\lg\sigma)$ bits.
This is very important in many practical applications where the data is so
highly compressible that a redundancy of $o(n\lg\sigma)$ bits would dominate 
the overall space. While there exist representations using even $n\Ho(s)+o(n)$ 
bits, ours is the first one supporting the operations in time
$\Oh{\lg\lg\sigma}$ while breaking the $o(n\lg\sigma)$ redundancy barrier.
Moreover, our time complexities are adaptive to the compressibility of the
sequence, reaching average times $\Oh{\lg \Ho(s)}$ under reasonable
assumptions. We have given various byproducts of the result, where the 
compressed-redundancy property carries over representations of text indexes,
permutations, functions, binary relations, and so on. It is likely that still 
other data structures can benefit from our compressed-redundancy representation.
Finally, we have shown experimentally that our representation is highly 
practical, on large alphabets, both in synthetic and real-life application 
scenarios.

\medskip

On the other hand, various interesting challenges on sequence representations
remain open:
\begin{enumerate}
\item Use $n\Hk(s)+o(n)(\Hk(s)+1)$ bits of space, rather than
$n\Hk(s)+o(n\lg \sigma)$ \cite{BHMR07,GOR10} or our
$n\Ho(s)+o(n)(\Ho(s)+1)$ bits, while still supporting the queries 
$\access$, $\rank$, and $\select$ efficiently. 
\item Remove the $o(n\Ho(s))$ term from the redundancy while retaining
loglogarithmic query times. Golynski et al.~\cite{GRR08} have achieved 
$n\Ho(s)+o(n)$ bits of space, but the time complexities are exponentially
higher on large alphabets, $\Oh{1+\frac{\lg\sigma}{\lg\lg n}}$.
\item Lower the $o(n)$ redundancy term, which may be not negligible on 
highly compressible sequences.
Our $o(n)$ redundancy is indeed $\Oh{\frac{n}{\lg\lg\lg n}}$.
That of Golynski et al.~\cite{GRR08}, $o\left(\frac{n\lg\sigma}{\lg n}\right)$,
is more attractive, at least for small alphabets. Moreover, for the binary 
case, P\u{a}tra\c{s}cu~\cite{Pat08} obtained $\Oh{\frac{n}{\lg^c n}}$ for any 
constant $c$, and this is likely to carry over multiary wavelet trees.
\end{enumerate}

After the publication of the conference version of this paper, Belazzougui
and Navarro \cite{BNesa11} achieved a different tradeoff for one of our
byproducts (Theorem~\ref{thm:self}). By spending $\Oh{n}$ further bits, they
completely removed the terms dependent on $\sigma$ in all time complexities,
achieving $\Oh{m}$, $\Oh{\lg n}$ and $\Oh{r-l+\lg n}$ times for counting,
locating and extracting, respectively. Their technique is based in monotone
minimum perfect hash functions (mmphfs), which can also be used to improve
some of our results on permutations, for example obtaining constant time for
query $\pi(i)$ in Theorem~\ref{thm:runs} and thus improving all the derived
results\footnote{Djamal Belazzougui, personal communication.}.

This is just one example of how lively current research is on this 
fundamental problem. Another
example is the large amount of recent work attempting to close the gap
between lower and upper bounds when taking into account compression, time
and redundancy
\cite{GGGRR07,Gol07,GRR08,Pat08,GORR09,Pat09,Gol09,GOR10}. Very recently,
Belazzougui and Navarro \cite{BNlower11} proved a lower bound of
$\Omega(\lg\frac{\lg\sigma}{\lg w})$ for operation $\rank$ on a RAM machine
of word size $w$, which holds for any space of the form $\Oh{nw^{\Oh{1}}}$, and
achieved this time within $\Oh{n\lg\sigma}$ bits of space. Then, making use
of the results we present in this paper, they reduced the space to
$n\Ho(s)+o(n\Ho(s))+o(n)$ bits. This is just one example of how our technique
can be easily used to move from linear-space to compressed-redundancy-space
sequence representations.

\paragraph*{Acknowledgments.}

We thank Djamal Belazzougui for helpful comments on a draft of this paper,
and Meg Gagie for righting our grammar.

\bibliographystyle{plain}
\bibliography{paper}

\end{document}